\documentclass[aps,pra,twocolumn,superscriptaddress,10pt]{revtex4-2}

\usepackage{graphicx,subfigure}
\usepackage{amsmath,amssymb,bm,amsthm}
\usepackage{mathtools}
\usepackage{multirow}
\usepackage{textcomp}
\usepackage{float}
\usepackage{xcolor}
\usepackage[normalem]{ulem}
\usepackage{dsfont}
\usepackage{url}
\usepackage{physics}
\usepackage{comment}
\usepackage{mathrsfs}
\usepackage{makecell}
\usepackage{array}
\usepackage{lineno}

\usepackage{tikz}
\usetikzlibrary{shapes, arrows, positioning}

\usepackage{hyperref}

\hypersetup{
    colorlinks,
    citecolor=black,
    filecolor=black,
    linkcolor=black,
    urlcolor=black
}

\newtheorem{theorem}{Theorem}
\newtheorem{lemma}{Lemma}
\usepackage{cleveref}

\begin{document}

\title{Equilibrium propagation in non-equilibrium quantum systems}

\author{Aditya Gandotra}
\email{adityag@uchicago.edu}
\author{David D. Noachtar}
\author{Aashish A. Clerk}
\email{aaclerk@uchicago.edu}
\affiliation{Chicago Quantum Institute, Pritzker School of Molecular Engineering, University of Chicago, Chicago, IL 60637}

\date{\today}

\begin{abstract}
Equilibrium propagation (EP) is a powerful training algorithm for self-learning physical systems that is widely used for a variety of machine learning inspired tasks. 
It relies on the existence of an energy function that depends on tunable training parameters, and can be extended to apply to many-body quantum systems in equilibrium. 
A key challenge however is the need to repeatedly relax to the system ground state after each parameter update. 
We show here a route for circumventing this problem by extending EP to many-body driven-dissipative quantum systems, where the dynamics naturally relaxes to a parameter-dependent non-equilibrium steady state. 
Despite the lack of a free energy, we find that for a class of open systems, a quantum detailed balance symmetry still allows for an efficient EP strategy.
This relies crucially on a new class of  non-equilibrium static susceptibility relations that we derive in detail.  
We analyze examples of our new non-equilibrium training paradigm for two experimentally-relevant models: a driven-dissipative transverse field Ising model, and a parametrically-driven bosonic lattice.
\end{abstract}

\maketitle

\section{Introduction}

Physical learning refers to the process by which a physical system learns some desired behavior by modifying internal degrees of freedom in response to some external stimuli~\cite{stern2023learning}.
These learning parameters are trained by a {\it learning rule}, which describes the dynamical process by which they are updated.
Training often proceeds using gradient descent, where the physical dynamics use this learning rule to minimize some cost function. 
A central question in this emerging field is how to efficiently extract gradients from state preparations and measurements alone~\cite{momeni2025training, lopez2023self}.
Equilibrium propagation (EP) is a powerful training algorithm that allows one to extract gradients in this way for systems in thermal equilibrium~\cite{scellier2017equilibrium}.
This method has been proposed and successfully implemented in a variety of classical systems~\cite{Scellier2021DeepLearningTheoryPhysics, Kendall2020EndToEndAnalogEqProp,Martin2021EqSpike,Ernoult2020ContinualWeightUpdates, Scellier2022AgnosticPhysicsDrivenDL, Falk2025TemporalContrastiveLearning,Wang2024CoupledPhaseOscillatorsEqProp, laydevant2024training, Stern2021SupervisedLearningPhysicalNetworks,Stern2024PhysicalEffectsOfLearning,OConnor2019SpikingEqProp,Dillavou2022DecentralizedPhysicsDrivenLearning}, including non-linear resistor networks~\cite{Dillavou2023MachineLearningWithoutAProcessor} and elastic networks~\cite{Altman2024CoupledLearningElasticNetworks}.

It is interesting to ask if EP can be used for physical learning with quantum systems.
Recent work shows this can be done for equilibrium quantum systems described by a Hamiltonian~\cite{wanjura2025quantum, massar2025equilibrium, scellier2024quantum}. 
Despite its many advantages, a potential drawback of EP (especially relevant for quantum systems) is its requirement to constantly re-equilibrate the system after each nudge of model parameters.
In the context of equilibrium quantum systems, EP can then only be performed through repeated preparation of ground or thermal states of many-body systems, for potentially thousands of training steps.
Preparing ground states of many-body Hamiltonians is QMA-hard in the worst case~\cite{kitaev2002classical, kempe2006complexity, oliveira2005complexity}, and while recent quantum Gibbs samplers (e.g.~\cite{chen2025efficient}) provably prepare thermal states efficiently in certain regimes (e.g.~at high temperatures, or for 1D systems)~\cite{bergamaschi2026fast, rouze2026efficient,rouze2026optimal, tong2025fast}, efficient preparation for generic Hamiltonians at low temperature remains out of reach.
This would seem to pose a serious challenge for large-scale quantum EP-based learning.

\begin{figure}
    \centering
    \includegraphics[width=\linewidth]{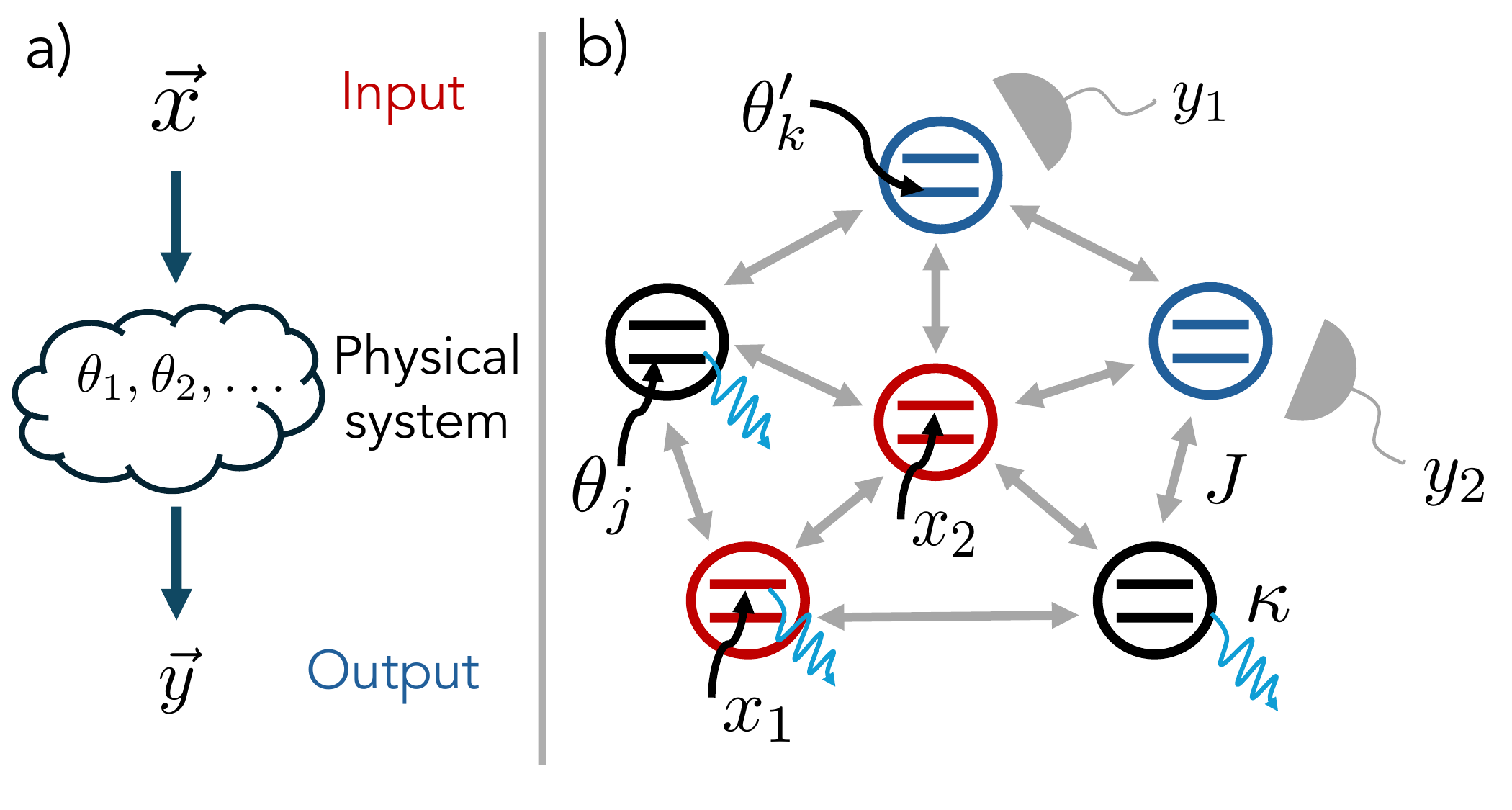}
    \caption{
    (a) A generic physical learning setup, where system parameters $\theta_j$ are trained so that the system reproduces a target input-output function from a set of parameters $\vec x$ to measured quantities $\vec y$.
    (b) Implementation in a driven-dissipative system of interacting qubits.  Internal parameters $\theta_j$ and inputs $x_k$ are drive amplitudes; the qubits interact with strength $J$ and are subject to decay at rate $\kappa$.  }
    \label{fig:schematic}
\end{figure}

In this work, we explore an alternative approach that addresses this crucial state-preparation issue: we extend quantum EP to a setting where equilibration dynamics is not assumed, but is instead built-in to the description of the system.
We study driven-dissipative quantum systems which relax to a non-equilibrium steady state that depends on system parameters, circumventing the problem of assuming an external means for repeated ground state cooling. 
The goal is to then use these parameter-dependent steady states for learning. 
Naively, it would seem that EP is now impossible:
our new setting produces parameter-dependent non-thermal states, whereas EP requires the existence of an equilibrium free energy.  
The question then becomes whether it is ever possible to salvage EP learning in truly non-equilibrium quantum settings.  

Surprisingly, the answer is yes. 
The key insight is that full thermal equilibrium is not necessary for `equilibrium' propagation. 
Instead, EP just requires a specific symmetry of static susceptibilities in the steady state.
While thermal equilibrium guarantees such symmetries, they can also occur robustly in more general non-equilibrium systems.
A key contribution of our work is to derive susceptibility symmetries and demonstrate EP-based training in driven-dissipative quantum systems which satisfy a version of quantum detailed balance known as \textit{hidden time-reversal symmetry} (hTRS)~\cite{roberts2021hidden}.
hTRS enables us to promote trivial gauge covariances in dissipative dynamics to a much more powerful set of relationships, which in turn guarantees susceptibility symmetries (see Fig.~\ref{fig:thms} for a sketch of the key idea).
Consider a driven-dissipative many-body system determined by a large set of drive parameters.
In many cases, a simple unitary transformation relates the dynamics and steady state corresponding to one choice of parameters to those corresponding to another choice of parameters (e.g. with different drive phases). 
In Fig.~\ref{fig:thms}, models related in this way all lie on a circle in an abstract parameter space. 
The central physical result of our work, contained in Theorems~\ref{theorem:tfim} and~\ref{theorem:bosonic} below, is that hTRS allows us to non-trivially extend this gauge covariance to relate models that differ both in phases {\it and magnitude} of drive amplitudes (see orange arrow in  Fig.~\ref{fig:thms}).
As we show, this non-trivial mapping between different steady states directly allows us to prove the crucial susceptibility symmetries that enable EP.

\begin{figure}
    \centering
    \includegraphics[width=0.75\linewidth]{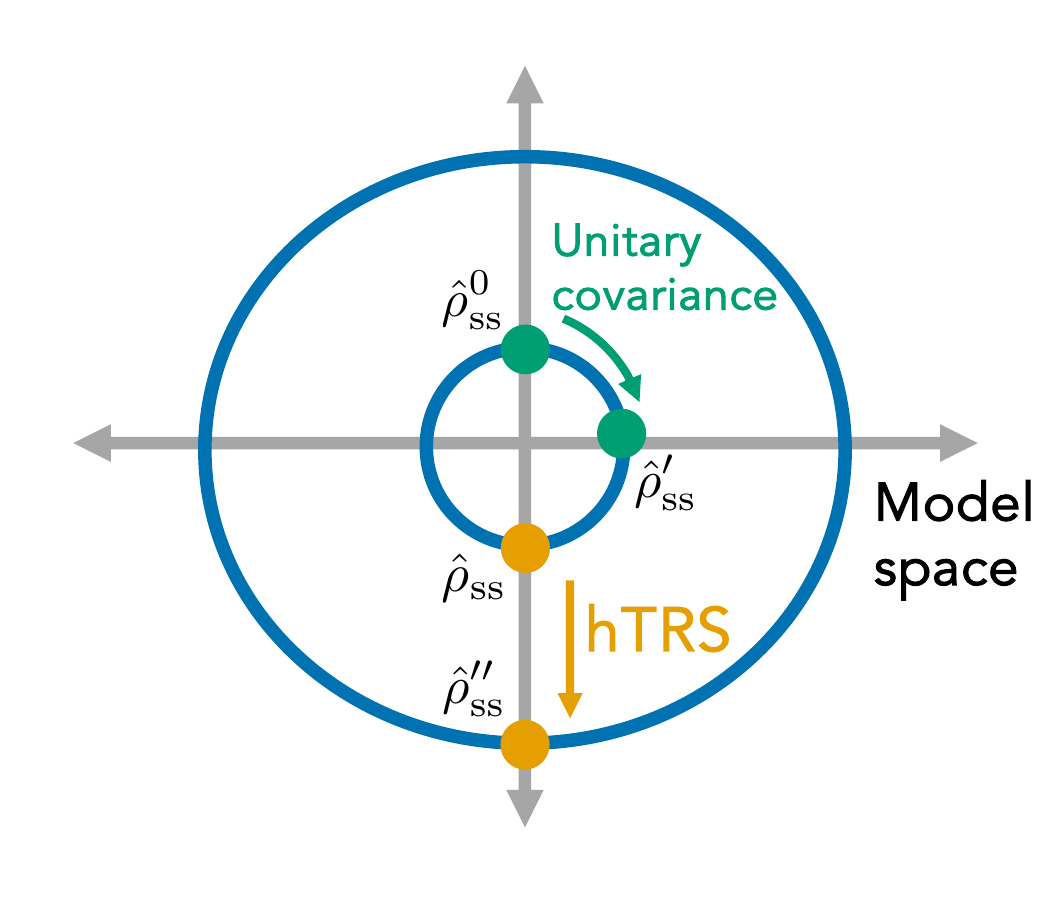}
    \caption{Schematic of the space of drive parameters for a driven-dissipative quantum system. A circle of a fixed radius schematically represents the set of models whose steady states are related by a unitary transformation, i.e.~\(\hat{\rho}'_{\rm ss} = \hat U \hat \rho^0_{\rm ss} \hat U^{\dag}\) for some unitary \(\hat U\). When the system has hTRS, this gauge covariance is effectively `complexified': we can write down an explicit relation between steady states that lie on circles of different radii, i.e.~for drives {\it anywhere} in the parameter space of the model. This non-trivial extension of the unitary covariance is the main content of Theorems~\ref{theorem:tfim} and~\ref{theorem:bosonic}.}
    \label{fig:thms}
\end{figure}

We demonstrate these new non-equilibrium susceptibility relations and EP based learning explicitly for two many-body dissipative systems that are experimentally relevant: a dissipative transverse-field Ising model (DTFIM)~\cite{roberts2023exact} (Theorem~\ref{theorem:tfim}), and an interacting bosonic lattice subject to two-photon driving~\cite{roberts2023competition} (Theorem~\ref{theorem:bosonic}).
These examples are meant to be illustrative but not exhaustive: there are many other many-body open systems that also have hTRS and could be interesting for EP learning (e.g.~dissipative spin chains~\cite{yao2025hidden,lingenfelter2024exact}, and interacting fermionic models~\cite{lingenfelter2026exact}).  

The rest of the paper is organized as follows. Given our goal of bridging topics in physical learning theory with those from many-body open quantum systems, we discuss background topics from these fields relevant to our work in Sections~\ref{section:closed-ep} to~\ref{section:htrs} (including EP and hidden time-reversal in open quantum systems). 
In Section~\ref{section:dtfim}, we turn to the main setting of our work. 
We detail the EP training procedure in the disordered, dissipative transverse-field Ising model (DTFIM), and demonstrate a simple example of training.
Then, in Section~\ref{section:bosonic-lattice}, we do a similar analysis on an interacting bosonic lattice model, presenting the susceptibility symmetry and a simple training example. 

We note that in classical systems, there has been some work to extend EP to general dynamical systems.
These studies generally pursue one of two directions, each fundamentally distinct from our work.
One strategy involves modifying the dynamics in each training step to make them more `equilibrium-like' such that one can express an effective free energy \cite{scellier2018generalization, scurria2026equilibrium}.
In contrast, we use hTRS to define a generating function for a {\it subset} of system observables we use for EP, without ever requiring global modifications of system dynamics. 
The other strategy, proposed in the context of driven-dissipative optical systems~\cite{cin2025training}, develops an EP-style two phase training approach reliant on a susceptibility symmetry being approximately correct for slightly non-linear systems.
In our work, the susceptibility symmetries are exact, and there is no linearity assumption: both models we study in this work are highly non-linear systems.

Finally, before proceeding, we wish to clarify the focus and scope of our work. 
When studying physical learning, there are two essential yet distinct questions that one can ask: that of {\it trainability} (i.e.~how efficiently does a training algorithm perform gradient descent?) and that of {\it expressivity} (i.e.~how good is a specific system at learning complex data?).
Our focus here is on a new training algorithm based on EP in non-equilibrium quantum systems; as such, our focus will be on the first question of trainability.  The question of expressivity of generic non-equilibrium systems with hTRS is of course interesting and important, but we leave a more thorough study of this question for future work.

\section{Background}
\subsection{Physical learning using  quantum ground states}
\label{section:closed-ep}

To set the context for our work, we first review the basics of how closed quantum systems can be used for learning, and how equilibrium propagation can be implemented in this context~\cite{wanjura2025quantum}.
The overarching goal is to use an input-output relation $\vec{y} = f(\vec{x},\vec{\theta})$ defined by a physical system with many internal parameters $\vec{\theta}$ to learn a given target function $\vec{y} = f_{\rm target}(\vec{x})$ using a set of training data.
Our system will be an $N$-site many-body quantum system where a subset of lattice sites will be used to encode input data $x_j$ via local Hermitian operators \(\hat{X}_j\).  The remaining lattice sites will be used as trainable parameters, and will each have an associated Hermitian operator \(\hat{O}_j\).  
The Hamiltonian is:
\begin{equation}
\label{eq:closed-hamiltonian}
    \hat{H}\left( \vec{x}, \vec{\theta}\right) = \sum_{j = 1}^{M_{\text{train}}} \theta_j \hat{O}_j +  \sum_{j = M_{\text{train}}+1}^{N} x_j \hat{X}_j + \hat{H}_{\text{int}}
\end{equation}
where \(x_j\) and \(\theta_j\) are tuneable local fields, and \(\hat{H}_{\text{int}}\) describes fixed interaction terms.  

As we now detail, expectation values in the parameter-dependent ground state of the system \(\ket{\Psi(\vec{x}, \vec{\theta})}\) will define the relevant input-output relation.
We partition the non-input sites \([1,\dots, M_{\text{train}}]\)  into a set of sites \([1,\dots, M_{\text{out}}]\) and \( [M_{\text{out}}+1,\dots, M_{\text{train}}]\) for some integer \(M_{\text{out}} \leq M_{\text{train}}\).
We denote the control fields and observables on the first set (the `output' set) of sites as
\begin{equation}
    \theta'_j \equiv \theta_j \quad \hat{O}'_j \equiv \hat{O}_j \quad  j \in [1, M_{\text{out}}].
\end{equation}
Similarly, we denote the control fields on the second set of sites (the `internal' sites) as
\begin{equation}
     \theta''_j \equiv \theta_j \quad \hat{O}''_j \equiv \hat{O}_j \quad  j \in [M_{\text{out}}+1, M_{\text{train}}].
\end{equation}
We use \(\vec{\theta}\) (\(\vec{O}\)) to refer to the vector of {\it all} control fields (observables) on sites \([1, M_{\text{train}}]\), \(\vec{\theta'}\) (\(\vec{O}'\)) to denote control fields (observables) for only the output sites \([1, M_{\text{out}}]\), and \(\vec{\theta''}\) (\(\vec{O}''\)) for control fields (observables) for only the internal sites \([M_{\text{out}}+1, M_{\text{train}}]\).
The output of the system can then be defined in terms of the ground state expectation values
\begin{equation}
    \vec{y} = \langle \Psi | \vec{O}' | \Psi \rangle.
\end{equation}
These outputs are a function of the parameter-dependent ground state, hence defining our trainable input-output relation  \(\vec{y} = f(\vec{x}, \vec{\theta})\).

We next need to specify how training data \((\vec{x}_{\text{train}}, \vec{y}_{\text{train}})\) generated by sampling from the target function will be used to train the internal parameters $\vec{\theta}$.
This can be achieved by defining an appropriate cost function \(C(\vec{y}, \vec{y}_{\text{train}})\) that is then minimized using gradient descent.
We will use throughout this work the mean squared error (MSE), \(C = 1/M_{\rm out}\sum_l\|y_l - y_l^{\text{train}} \|^2\), where \(M_{\rm out}\) is the size of the output vector \(\vec y\). 
Training occurs iteratively, where one first computes the gradient \(\partial C/\partial \theta_j\) for the current choice of parameters; as we discuss, there are different ways to do this.
This gradient is then used to update the internal training parameters in each training step via
\begin{equation}
\label{eq:training-update}
    \theta_j \to \theta_j - \eta\frac{\partial C}{\partial \theta_j}
\end{equation}
where \(\eta\) is a tunable hyperparameter referred to as the learning rate. 

The gradient of the cost function can be written as
\begin{equation}
\label{eq:cost-func-dotproduct}
    \frac{\partial C}{\partial \theta_j} = \sum_k \frac{\partial C}{\partial y_k}\frac{\partial y_k}{\partial \theta_j} = \vec{\varepsilon}\cdot \frac{\partial\langle\vec{O}'\rangle}{\partial \theta_j}
\end{equation}
where we define
\begin{equation}
\label{eq:error-vec}
    \vec{\varepsilon} \equiv \underbrace{\left(\frac{\partial C}{\partial y_1},\dots,\frac{\partial C}{\partial y_{M_{\text{out}}}}\right)}_{M_{\text{out}}\text{ output sites}}.
\end{equation}
This is an \textit{error vector} that encodes all the information about the cost function, and is defined only on the output sites.
The question now becomes how to efficiently compute the needed gradients of observables \( \vec{y} = \langle \vec{O}' \rangle\).

\subsubsection{Parameter-shift method}
\label{section:param-shift}
The most direct approach to computing the needed gradients is known as the 
\textit{parameter-shift} method.
It involves using a finite difference approximation for the derivatives in Eq.~(\ref{eq:cost-func-dotproduct}).
To this end, we take measurements of the output observables \(\vec{y}\) in two ground states.
First, we prepare a `free' state, i.e.~the ground state corresponding to our current choice of system parameters:
\begin{equation}
    \ket{\Psi_0} \equiv \ket{\Psi(\vec{x},\vec{\theta})}.
\end{equation}
The second state is the `nudged' ground state, 
i.e.~the ground state corresponding to a new set of system parameters where just one of the training parameters \(\theta_j\) has been changed from the ``current" parameter values by a small amount \(\beta\).  We can write the new set of parameters as \(\vec{\theta} \to \vec{\theta} + \beta \vec{n}_j\) where \(\vec{n}_j\) is a unit vector defined as 
\begin{equation}
    \label{eq:unit-vector}
    \vec{n}_j \equiv (0,\dots\underbrace{1}_{j\text{th site}},\dots,0).
\end{equation}
The nudged ground state is then defined as
\begin{equation}
    \ket{\Psi_j} \equiv \ket{\Psi(\vec{x}, \vec{\theta} + \beta\vec{n}_j)}.
\end{equation}
We can then approximate the gradient of the cost function as
\begin{equation}
\label{eq:grad-param-shift}
    \frac{\partial C}{\partial \theta_j} \approx  \vec{\varepsilon} \cdot \frac{\langle \Psi_j|\vec{O}'|\Psi_j\rangle -  \langle \Psi_0|\vec{O}'|\Psi_0\rangle}{\beta}.
\end{equation}

Each of the needed gradients is computed using a three-step procedure:
\begin{itemize}
    \item Measure output observables in present state to obtain \(\langle \Psi_0|\vec{O}'|\Psi_0\rangle\).
    \item Perturb a single training parameter \(\theta_j \to \theta_j + \beta\).
    \item Prepare the system in a new ground state \(\ket{\Psi_j}\), and then again measure the output observables to obtain \(\langle \Psi_j|\vec{O}'|\Psi_j\rangle\).
\end{itemize}
The gradient is then computed as in Eq.~(\ref{eq:grad-param-shift}). 
It is important to note that we must repeat this procedure to get the gradient for each of the training parameters: a new nudged ground state \(\ket{\Psi_j}\) must be prepared for each \(\theta_j\). 
Therefore, for a single gradient update, we need to prepare \(M_{\text{train}}+1\) independent ground states (one free state and a nudged state for each training parameter), and in each of them perform \(M_{\text{out}}\) measurements.  
This results in a total of \(M_{\text{out}} (M_{\text{train}}+1)\) measurements per training update.
One of course needs to repeat each of these measurements several times to faithfully estimate the expectation values, as is standard in any quantum experiment.
In many cases, the limiting aspect of this procedure will be the large number of state preparations required of different ground states, across potentially tens of thousands of training steps.
Preparing many-body systems in low-temperature states is generically difficult, both from a practical experimental perspective, but also on fundamental grounds.
All these reasons motivate the search for better methods for extracting gradients.

\subsubsection{Equilibrium propagation}
\label{section:closed-ep-2}
Equilibrium propagation (EP)~\cite{scellier2017equilibrium,Scellier2021DeepLearningTheoryPhysics, Kendall2020EndToEndAnalogEqProp,Martin2021EqSpike,Ernoult2020ContinualWeightUpdates, Scellier2022AgnosticPhysicsDrivenDL, Falk2025TemporalContrastiveLearning,Wang2024CoupledPhaseOscillatorsEqProp, laydevant2024training, Stern2021SupervisedLearningPhysicalNetworks,Stern2024PhysicalEffectsOfLearning,OConnor2019SpikingEqProp,Dillavou2022DecentralizedPhysicsDrivenLearning} is an alternative to the parameter shift method that offers a remarkable advantage: by exploiting the fact that the system state is governed by a parameterized energy (or free energy), it extracts gradients with a dramatically smaller number of independent state preparations.   

In the first phase of the EP training procedure (the ``free" phase) we prepare the ground state \(\ket{\Psi_0}\) corresponding to the current set of system parameters, and then measure {\it all} \(M_{\text{train}}\) observables $\hat{O}_j$ in this state (i.e.~both on the internal sites and output sites) obtaining \(\langle \Psi_0|\vec{O}|\Psi_0\rangle\).
This is in contrast to the parameter-shift method, where we must only measure the output observables \(\hat{O}_j'\) in both phases.
In the second ``nudged" phase, we prepare a {\it single} new perturbed ground state. This new ground state corresponds to {\it simultaneously nudging all the local fields on the output sites}, with a nudge that is proportional to the size of the error vector:  \(\vec{\theta'} \to \vec{\theta'} + \beta \vec{\varepsilon}\). This yields the new ground state
\begin{equation}
\label{eq:ep-nudged-close}
    \ket{\Psi_0'} \equiv \ket{\Psi(\vec{x}, \vec{\theta'} + \beta \vec{\varepsilon}, \vec{\theta''})}.
\end{equation} 
We once again measure {\it all} observables $\hat{O}_j$ in this new state, obtaining \(\langle \Psi'_0|\vec{O}|\Psi'_0\rangle\).
The above measurement results are enough to calculate all the needed gradients via (see App.~\ref{section:EP-proof} for a quick proof):
\begin{equation}
\label{eq:grad-ep-closed}
     \frac{\partial C}{\partial \theta_j} \approx   \frac{\langle \Psi'_0|\hat{O}_j|\Psi'_0\rangle -  \langle \Psi_0|\hat{O}_j|\Psi_0\rangle}{\beta}.
\end{equation}
The form of Eq.~\eqref{eq:grad-ep-closed} directly shows the advantage of EP over the parameter shift method: 
all training parameters are updated simultaneously (i.e.~using a single nudged state $|\Psi'_0\rangle$), and one does not require a new nudged state for each training parameter. Hence, in EP, only two states are needed per gradient update (\(\ket{\Psi_0}, \ket{\Psi_0'}\)), independent of the number of parameters being trained.
Further, we perform only \(2M_{\text{train}}\) total measurements per training update.
In this way, EP scales more favorably than the parameter shift method in terms of total state preparations, and when \(M_{\rm out} \geq 2\) also in terms of total measurements.

EP and Eq.~\eqref{eq:grad-ep-closed} crucially rely on the fact that we are always using states that minimize an energy function (or free energy). Specifically, EP requires that parameter-dependent expectation values obey a susceptibility symmetry, namely:
\begin{equation}
    \label{eq:reciprocity-closedquantum}
    \frac{\partial \langle \hat{O}_i \rangle}{\partial \theta_j} = \frac{\partial \langle \hat{O}_j \rangle}{\partial \theta_i}.
\end{equation}
In App.~\ref{section:EP-proof}, we show how this susceptibility symmetry is sufficient to prove Eq.~\eqref{eq:grad-ep-closed}.
Note that the expectation values here each depend on the parameters $\vec \theta$, as they are taken in $\vec \theta$-dependent ground states.

We emphasize that this relation is formally distinct from Onsager reciprocity (see e.g. Ref.~\onlinecite{pathria2017statistical}), as it does not rely on any notion of time-reversal symmetry.  Instead, it is a direct consequence of the existence of an energy function.  
To be explicit, the parameter-dependent ground state energy of our system
$\mathcal{F}(\vec \theta)$
can be written as a ground-state expectation value of the system Hamiltonian:
\begin{equation}
    \mathcal{F}(\vec \theta) = \langle \Psi_0(\vec \theta) | \hat{H}(\vec \theta) | \Psi_0(\vec \theta) \rangle.
\end{equation}
Given the form of our Hamiltonian in Eq.~\eqref{eq:closed-hamiltonian}, it thus follows that we can express expectation values of an observable $\hat{O}_i$ for some parameters $\vec \theta_0$ in terms of a derivative with respect to its conjugate field:
\begin{equation}
\label{eq:free-energy-derivative}
    \langle\hat{O}_i\rangle = \frac{\partial \mathcal{F}}{\partial \theta_i},
\end{equation}
where both sides should be evaluated at $\vec \theta = \vec \theta_0$.  This is a direct consequence of the well-known Hellmann-Feynman theorem~\cite{feynman1939forces}.  
The susceptibility symmetry Eq.~\eqref{eq:reciprocity-closedquantum} is then a simple consequence of the symmetry of second derivatives (sometimes referred to as Schwarz's theorem):
\begin{equation}
    \frac{\partial^2 \mathcal{F}}{\partial \theta_j \partial \theta_i} = \frac{\partial^2 \mathcal{F}}{\partial \theta_i \partial \theta_j} \implies  \frac{ \partial \langle \hat{O}_i \rangle}{\partial \theta_j} = \frac{ \partial \langle \hat{O}_j \rangle}{\partial \theta_i}.
\end{equation}

We thus see that the EP training protocol will work in any setting where we are working with parameter-dependent states that minimize a parameter-dependent energy or free energy; this is the standard setting for EP \cite{scellier2017equilibrium}.  This might give the impression that once we move away from equilibrium systems, EP will not be possible, as one is no longer guaranteed to have a meaningful notion of a free energy.  This is overly pessimistic.  We note that all EP requires is the effective symmetry of susceptibilities expressed in Eq.~\eqref{eq:reciprocity-closedquantum}.  We will show that such relations can hold even in non-equilibrium quantum settings without a well-defined free energy.  Instead, these relations can result as the consequence of generalized notions of Onsager symmetry and quantum detailed balance \cite{roberts2021hidden}.  

\subsection{Physical learning using non-equilibrium steady states of driven-dissipative quantum systems}
\label{section:learning-problem}

We now consider the driven-dissipative analogue of the closed-system learning problem introduced in Sec.~\ref{section:closed-ep}.
Consider a quantum many-body open system whose dynamics are Markovian and governed by a time-independent generator.
The system is described by a density  matrix $\hat{\rho}$, whose dynamics follows a Gorini-Kossakowski-Sudarshan-Lindblad (GKSL) form master
equation \cite{lindblad1976generators, gorini1976completely} ($\hbar = 1$ throughout):
\begin{equation}
\label{eq:master-eq}
     \partial_t \hat{\rho} = -i[\hat{H},\hat{\rho}] + \sum_k \mathcal{D}[\hat{c}_k] \hat{\rho} \equiv\mathcal{L}[\hat{\rho}].
\end{equation}
Here the Lindblad dissipative superoperator is defined as \(\mathcal{D}[\hat{c}] \hat{\rho} \equiv \hat{c}\hat{\rho}\hat{c}^{\dag} - \frac{1}{2} \{\hat{c}^{\dag}\hat{c}, \hat{\rho}\}\).
The first term on the RHS of Eq.~\eqref{eq:master-eq} describes purely Hamiltonian dynamics (which by itself would generate unitary evolution), while the second term describes both dissipation and incoherent driving. Each dissipative term is parameterized by a ``jump operator" \(\hat{c}_k\); these operators describe how the system is coupled to the dissipative environment.  We will be interested in the steady-state of the dynamics (often termed a non-equilibrium steady state (NESS)), which we denote \(\hat{\rho}_{\rm ss}\).  Such a steady state satisfies: 
\begin{equation}
    \mathcal{L}[\hat{\rho}_{\rm ss}] = 0.
\end{equation}
In this work, we restrict attention to cases where for any fixed choice of parameters, there is a unique NESS.  

For the open-system version of the learning problem in Sec.~\ref{section:closed-ep}, we consider a setup governed by a parameter-dependent master equation of the form of Eq.~\eqref{eq:master-eq}. 
For simplicity, we assume that the jump operators $\hat{c}_k$ are parameter-independent, whereas the system Hamiltonian is a function of both input variables $\vec x$ and training parameters $\vec \theta $, i.e.~\(\hat H = \hat{H}(\vec{x}, \vec{\theta})\). 
This then yields a parameter dependent NESS:
\begin{equation}
    \hat{\rho}_{\rm ss} = \hat{\rho}_{\rm ss}(\vec{x}, \vec{\theta}).
\end{equation}
Outputs now correspond to the expectation value of output-site observables \(\langle \hat{O}_j \rangle\) in this steady state, and we have an analogous input-output relation to the closed system.
Critically, in this case, these observables need not be the ones conjugate to the input drives \(\vec x\); indeed, as we will show, hTRS requires that these observables in fact do not appear in the Hamiltonian at all.

This general setup gives us a parameter-dependent state and a parameter dependent input-output relation, just like the Hamiltonian ground state case.  We could thus directly employ the parameter-shift method to perform learning in this setting.  Note that even at this level, there is in principle an advantage over the ground state learning protocol of Sec.~\ref{section:closed-ep}.  There, one had to assume some mechanism for cooling to the new ground state every time parameters were modified. In contrast, by specifying a full dissipative master equation with a unique steady state, we are defining both, a family of parameter-dependent steady states, and also a means for preparing these steady states (i.e.~the dissipative dynamics directly leads to the NESS at long times, no other external state preparation machinery or ``fridge" is required). Given the general difficulty of preparing ground states of many-body Hamiltonians, this already represents a potentially large advantage.    

While the parameter-shift method works in this open-systems setting without modification, at first glance there is no way to employ an EP strategy: in general, $\hat{\rho}_{\rm ss}(\vec x, \vec \theta)$ does not correspond to the minimum of a well-defined free energy functional (see e.g.~\cite{fazio2025many}). There is thus no built-in structure that guarantees that the crucial susceptibility symmetry required for EP (c.f.~Eq.~\eqref{eq:reciprocity-closedquantum}) will hold. A key result of our work is to show that all hope is not lost:  there is a class of open quantum systems where Eq.~\eqref{eq:reciprocity-closedquantum} holds despite not having any well-defined notion of a free energy. The crucial ingredient, as we now explain, is a suitable notion of quantum detailed balance.

\subsection{Quantum detailed balance in dissipative systems}
\label{section:htrs}

To employ an EP strategy for physical learning with open quantum systems, the simplest approach would be to restrict oneself to Lindblad master equations whose steady states {\it always} correspond to thermal states of the system Hamiltonian (e.g.~through the use of Davies generators \cite{davies1974markovian, davies1976markovian, alicki1976detailed, kossakowski1977quantum}). While formally possible, constructing the required dissipative dynamics requires highly non-local, parameter-dependent jump operators $\hat{c}_k$ (see e.g.~\cite{guo2025designing}), leading to an experimentally infeasible setup. Further, in this case there is nothing truly new in the open systems setting, as one is again just using a thermal equilibrium state with a well-defined free energy for learning.

We instead take a different approach: we work with experimentally-feasible dissipative dynamics (i.e.~based on simple local dissipators) that lead to NESS that are not simply Gibbs states of the Hamiltonian. To maintain this richness while still being able to employ the remarkable advantages of an EP training approach, we focus on dissipative dynamics that obey a specific kind of quantum detailed balance, so-called hidden time-reversal symmetry (hTRS)~\cite{roberts2021hidden}.  Unlike simpler versions of quantum detailed balance, this symmetry can hold even when the NESS is very far from being a Gibbs state. We will show below that hTRS can enable exactly the kind of susceptibility symmetry in Eq.~\eqref{eq:reciprocity-closedquantum} needed for EP-based learning. 
This hTRS can be viewed as a modified version of the more familiar KMS detailed balance, sometimes referred to as the SQDB-\(\theta\) condition (where the modification introduces a non-trivial anti-unitary operator that plays the role of time-reversal)~\cite{fagnola2007generators, fagnola2010generators}. A detailed and formal discussion is given in \cite{roberts2021hidden, noachtar_exact_2026}; here we just state the basic definition and its immediate consequences.  

Recall that the simplest kind of classical detailed balance implies that the steady state is constrained by more than just the requirement of global stationarity: there are additional ``local" constraints on the dynamics, namely that the probability current between any pair of states must vanish in the steady-state.
In a similar fashion, hTRS is a form of quantum detailed balance that places independent constraints on each element of the dynamical generator. 
An open quantum system described by a Lindbladian Eq.~\eqref{eq:master-eq} has its dynamics generated both by dissipative terms corresponding to the jump operators $\hat{c}_k$, and also via the non-Hermitian Hamiltonian \(\hat{H}_{\text{eff}} = \hat{H} - (i/2)\sum_l \hat{c}_l^{\dag} \hat{c}_l\).
For the system to have hTRS, each of these is constrained:  there must exist an anti-linear operator $\hat \Psi$ \footnote{Recall that an anti-linear operator $\hat \Psi$ satisfies \(c \hat \Psi = \hat \Psi c^*\) for any complex scalar \(c\).} 
such that
\begin{align}
\label{eq:htrs-heff}
    \hat{H}_{\text{eff}} & = \hat{\Psi}  \hat{H}^{\dag}_{\text{eff}} \hat{\Psi}^{-1} \\
    \label{eq:htrs-ck}
    \hat{c}_k & = \hat{\Psi} \hat{c}^{\dag}_k \hat{\Psi}^{-1} \quad \forall k.
\end{align}
If these conditions hold, it immediately follows that $\hat \rho_{\rm ss} = \hat \Psi \hat \Psi^\dagger$ is a steady state of the Lindbladian. 

The existence of hTRS has some immediately relevant consequences.
In particular, any open quantum system with hTRS obeys a kind of Onsager symmetry that relates {\it certain} steady-state correlation functions.
Any Lindblad system with hTRS necessarily satisfies:
\begin{align}
\label{eq:htrs-symmetries}
    \langle \hat{c}_i(t) \hat{c}_j(0)\rangle & =  \langle \hat{c}_j(t) \hat{c}_i(0)\rangle \quad t>0 \quad  \forall i,j \\
    \langle \hat{H}_{\text{eff}}(t) \hat{c}_i(0)\rangle & =  \langle \hat{c}_i(t) \hat{H}_{\text{eff}}(0)\rangle \quad t >0 \quad \forall i.
\end{align}
Here, all expectation values are taken in the NESS. We stress that only this subset of correlation functions obey the above Onsager-like symmetry, as opposed to all correlation functions. As discussed extensively in Ref.~\cite{roberts2021hidden}, this means that open systems with hTRS can still have highly non-thermal NESS, e.g.~NESS that do not commute with the system Hamiltonian.  Surprisingly, there are a number of many-body open systems that exhibit hTRS, including dissipative collective spin models \cite{roberts2023exact}, interacting bosonic models \cite{roberts2023competition}, and driven-dissipative spin chain systems \cite{lingenfelter2024exact,yao2025hidden}.
In what follows, we study two of these models as paradigmatic examples, and show that they can indeed be used for EP-based learning.
As we will see, hTRS in these models directly enables the crucial susceptibility symmetry Eq.~\eqref{eq:reciprocity-closedquantum} needed for EP-based learning.  

\section{EP-based learning in a dissipative transverse field Ising model (DTFIM)}
\label{section:dtfim}

We now turn to the main focus of this work: a detailed study of a specific many-body open quantum system whose parameterized NESS can be used for EP-based learning despite the lack of a meaningful free energy. We will explore a dissipative, disordered transverse-field Ising model (DTFIM).  Ref.~\cite{roberts2023exact} showed that this system has hTRS, and used this symmetry to derive an exact analytic solution for the steady state (something that will allow us to simulate large systems, see App.~\ref{section:mps}). We now introduce the model and show how it can be used to enable EP-based learning. 

\subsection{DTFIM Lindbladian and susceptibility symmetry}
\label{section:dtfim-reciprocity}

The DTFIM we study describes $N$ Rabi-driven qubits 
(Pauli matrices \(\hat{\sigma}^j_{x,y,z}\)) that have spatially-varying drive amplitudes, and which interact via all-to-all Ising interactions. All qubits are driven at the same frequency, and assumed to have the same splitting frequency. Working in a rotating frame set by the drive, the system Hamiltonian is:  
\begin{equation}
\label{eq:dtfim-ham}
        \hat{H}_{\text{TFIM}} = \frac{J}{4N} \sum_{i\neq j} \hat{\sigma}_z^i \hat{\sigma}_z^j + 
        \frac{1}{2} \sum_{j=1}^N \left(
        \Delta \hat{\sigma}_z^j +
        \Omega_j \hat{\sigma}_x^j
        \right).
\end{equation}
Here $J$ is the Ising interaction amplitude, $\Delta$ is the drive detuning, and $\Omega_j$ is the drive amplitude on qubit $j$ (all taken without loss of generality to be real).    
Each qubit is subject to local relaxation at a rate $\kappa$, with corresponding jump operators 
\(\hat{c}_k = \sqrt{\kappa}\ \hat{\sigma}_-^k\). The full dynamics of the system's density matrix is then given by the Lindbladian:
\begin{equation}
\label{eq:tfim-lindbladian}
    \mathcal{L}_{\text{TFIM}}[\hat{\rho}] = -i[\hat{H}_{\text{TFIM}},\hat{\rho}] +\kappa\sum_j \mathcal{D}[\hat{\sigma}^j_-] \hat{\rho}.
\end{equation}
This model features a number of tuneable parameters that grow linearly with the number of qubits $N$, namely the real drive amplitudes $\Omega_j$ (see Fig.~\ref{fig:schematic}(b) for a schematic representation of the DTFIM). This allows our system to be used for a physical learning application.  

As discussed extensively in Refs.~\cite{roberts2023exact,marcuzzi2014universal,tucker_facilitating_2020,song_dissipation-induced_2025,xiang_switching_2025,leppenen_quantum_2024}, the NESS of this model has a rich structure stemming from the competition between driving, interactions and loss.   
In the case of homogeneous driving (\(\Omega_j \equiv \Omega\)) and for large $N$, the model is known to exhibit phase transitions between a completely depolarized state (where the driving dominates) and a polarized state (where the dissipation dominates)~\cite{roberts2023exact, noachtar_exact_2026}.
The obvious order parameter here is the average polarization, i.e.~$1/N\sum_j \langle \hat \sigma^j_z \rangle$. We will instead focus on another set of observables that are also extremely sensitive to the dynamics: the steady state coherence of each qubit, as measured by 
\(\langle \hat{\sigma}_y^j \rangle\). These are zero both in the phase where the dissipation dominates (i.e.~all qubits unexcited), and when the drive dominates (i.e.~all qubits depolarized).  
They are however non-zero and maximal when the system is tuned close to its phase transition.  

Remarkably, the steady-state $\hat \sigma_y$ coherences satisfy exactly the kind of susceptibility symmetry that enables EP-style learning.  Using hTRS, we can show that for any set of parameters, and for any pair of qubits $(i,j)$, steady-state expectation values obey:
\begin{equation}
    \label{eq:tfim-reciprocity}
    \frac{\partial \langle \hat{\sigma}_y^i\rangle}{\partial \Omega_j} = \frac{\partial \langle \hat{\sigma}_y^j\rangle}{\partial \Omega_i} \quad \forall i,j.
\end{equation}
This susceptibility symmetry is a central result of our work, and the proof is given in App.~\ref{section:dtfim-reciprocity-derivation}.
It arises from hTRS promoting a simple gauge covariance of the Lindbladian to a much more non-trivial symmetry that relates steady states of different instances of the DTFIM.
Specifically, when we transform the Lindbladian using a unitary \(\hat U = \exp\left(i/2 \sum_j \varphi_j \hat \sigma_z^j \right)\), it is mapped to a different DTFIM Lindbladian where drive phases are shifted; the steady state transforms in a corresponding way.
This covariance of the Lindbladian under a \(U(1)\) gauge transformation is trivially true, and does not require hTRS.
In the context of Fig.~\ref{fig:thms}, the `model space' is precisely the complex plane of drive amplitudes, and the gauge covariance is simply the freedom to move along a circle of a fixed radius (i.e.~relate steady states of DTFIMs with different drive phases).
Critically however, as we show in App.~\ref{section:dtfim-reciprocity-derivation}, hTRS turns this into a much more powerful statement, relating steady states of DTFIM Lindbladians with drives differing in both phase {\it and magnitude}.
This in turn enables us to prove Eq.~\eqref{eq:tfim-reciprocity}.

We stress that while Eq.~\eqref{eq:tfim-reciprocity} appears analogous to Eq.~\eqref{eq:reciprocity-closedquantum} for a closed system, there are crucial differences. 
First, notice that the observable \(\hat \sigma_y\) does not appear in the DTFIM Hamiltonian.
It is in fact crucial for the symmetry to hold that this observable is orthogonal in the transverse plane to the driving term in the Hamiltonian.
Second, this susceptibility symmetry {\it does not hold} for all possible choices of qubit observables. 
Both these features are just two ways the non-equilibrium nature of our system manifests itself.  
We now discuss how the above symmetry relation is used as the basis for an EP training protocol.

\subsection{EP-based training using non-equilibrium steady states of the DTFIM}
\label{section:ep-dtfim}
The training setup for EP in the DTFIM is analogous to the one described in Sec.~\ref{section:closed-ep} for closed systems.
We partition the drive amplitudes \(\Omega_j\) as:
\begin{equation}
\label{eq:tfim-partition}
    \Omega_j \equiv 
    \begin{cases}
        \theta'_j &  j \in[1, M_{\text{out}}] \\
        \theta''_j & j \in [M_{\text{out}}+1, M_{\text{train}}] \\
        x_j & j \in  [M_{\text{train}}+1, N].
    \end{cases}
\end{equation}
The function input is encoded in the drives \(x_j\) of the input spins \([M_{\text{train}}+1, N]\).
The output is defined as the $\hat{\sigma}_y$ expectation values
\begin{equation}
    y_j = \langle \hat{\sigma}_y^j\rangle, \quad j \in [1, M_{\text{out}}]
\end{equation}
for some integer \(M_{\text{out}} \leq M_{\text{train}}\), giving us the required input-output relation.
We refer to the spins \([1, M_{\text{out}}]\) as the output spins, and the spins \([M_{\text{out}}+1, M_{\text{train}}]\) as internal spins.
As in the Hamiltonian case, we use \(\vec{\theta}\) to collectively denote the drives on all the output and internal spins, \([1, M_{\text{train}}]\).
These \(\theta_j\) are the trainable parameters that are learned through EP.

Training using EP here again involves the same two-phase approach discussed in Sec.~\ref{section:closed-ep}, with each phase using a different NESS.  
In the first free phase, the system is allowed to relax to the NESS corresponding to the current set of system parameters, defined as
\begin{equation}
\label{eq:ness-free}
    \hat{\rho}_0 \equiv \hat{\rho}_{\rm ss}(\vec{x}, \vec{\theta}).
\end{equation}
In this state, we measure the observables \(\langle \hat{\sigma}_y^j \rangle\) on all \(M_{\text{train}}\) trainable spins.
In the second nudged phase, we now {\it simultaneously} alter all of the drive amplitudes of the output spins, with the nudge of each drive taken to be proportional to the size of the
error vector (c.f.~Eq.~\eqref{eq:error-vec}): \(\vec{\theta'} \to \vec{\theta'} + \beta \vec{\varepsilon}\).  The system then relaxes to the new steady state corresponding to these parameters, yielding the `nudged' NESS
\begin{equation}
\label{eq:ness-nudged}
    \hat{\rho}_\beta \equiv \hat{\rho}_{\rm ss}\left(\vec{x}, \vec{\theta'} + \beta \vec{\varepsilon}, \vec{\theta''}\right).
\end{equation}
We once again measure the observables \(\langle \hat{\sigma}_y^j \rangle\) on all internal and output spins.
The EP gradient update is then given by 
\begin{equation}
   \frac{\partial C}{\partial \theta_j} \approx \frac{\Tr(\hat{\sigma}_y^j \hat{\rho}_{\beta})-\Tr( \hat{\sigma}_y^j \hat{\rho}_0)}{\beta}.
\end{equation}
This follows from Eq.~\eqref{eq:tfim-reciprocity} via the same argument used for the closed system case, 
c.f.~App.~\ref{section:EP-proof}.  

We thus see that at a practical level, EP is implemented in our DTFIM model analogously to its implementation in a closed system (c.f.~Sec.~\ref{section:closed-ep}). The key difference now however is that we are using a NESS which is not determined by a free energy.  Just like the closed-system case, our quantum dissipative version of EP uses a nudge of all training parameters to calculate all gradient updates simultaneously.

In what follows, we use a `symmetric' nudging procedure, as is now standard across implementations of EP (see e.g.~\cite{wanjura2025quantum}).
This refers to performing two nudges of the output drives, of size \(\beta\), in opposite directions for better numerical convergence.
Explicitly, the gradient update is then given by
\begin{equation}
   \frac{\partial C}{\partial \theta_j} \approx \frac{\Tr(\hat{\sigma}_y^j \hat{\rho}_{\beta})-\Tr( \hat{\sigma}_y^j \hat{\rho}_{-\beta})}{2\beta}.
\end{equation}
We emphasize that in this case, symmetric EP needs {\it three} state preparations per gradient update; the first free state \(\hat \rho_{0}\) to determine the error vector in the current state, along with two nudged states \(\hat \rho_{\beta}\) and \(\hat \rho_{-\beta}\) for the gradient update.

In the simulations that follow, we also use symmetric nudging for parameter-shift; the total number of state preparations per gradient step for this method is then \(2 M_{\rm train}\).
Each training parameter gets two nudges, using a finite difference to estimate \(\partial \langle \vec O \rangle/\partial \theta_j\) (see Eq.~\eqref{eq:grad-param-shift}). Instead of a shared free state (as described in Sec.~\ref{section:param-shift}), each parameter needs two independent states to estimate the gradient, yielding \(2 M_{\rm train}\) states per gradient update as opposed to \(M_{\rm train} + 1\) in the one-sided nudging case.


\subsection{Teacher-student learning}
\label{section:teacher-student}

We now turn to a concrete demonstration of EP-based training of the DTFIM setup.
As stated in the introduction, our goal is to show the basic efficacy of this training approach, and not conflate this issue with the distinct question of expressivity.
To this end, we discuss teacher-student learning, a kind of supervised learning task where we know a priori that the target function is learnable. 
Teacher-student learning involves training one initialization of the system (the `student') using training data generated from a different initialization of the system with different parameters and/or architecture (the `teacher').
It has been used extensively in the machine learning literature as a way to train small, simple neural networks from more complicated ones 
(see e.g.~\cite{hinton2015distilling}).

We first construct the teacher system that will produce the target input-output relation.  This will be an \(N\) spin DTFIM where we just have a single input spin (with a drive amplitude $x$), and $M_{\rm out} > 1$ output spins.  The drive amplitudes on the non-input spins are fixed at the values \(\vec{\theta}_{\rm teacher}\) (see Table~\ref{table:dtfim-1} for exact values), and we take 
\(J/\kappa = 5, \Delta = 0\).
To produce the target input-output relation, for each value of $x$ in the input domain, we measure the outputs \(y^{\text{train}}_j  \equiv \langle \hat{\sigma}_y^j\rangle\) for \(j \in [1, M_{\text{out}}]\) in the corresponding NESS.  This generates our training data \((x,\vec{y}^{\text{train}})\) that we will attempt to learn with a second system.  

We then initialize a new \(N\)-spin DTFIM (the `student') with all the same system parameters as the teacher (\(J/\kappa = 5, \Delta =0\)) except the training drives, which we set to some random values \(\theta_j^{\text{student}}\).
The goal then is to train these student drives using the training data to learn the teacher input-output relation.
Since the systems are identical, this is in some ways a trivial task, since we know that a solution always exists.
This kind of setup is still useful, as it lets us benchmark a particular training strategy (in our case EP) both in terms of accuracy (i.e.~the final converged error) and efficiency (i.e.~how quickly is the student able to learn the internal parameters needed to mimic the teacher).   

\begin{figure}
    \centering
    \includegraphics[width=1.0\linewidth]{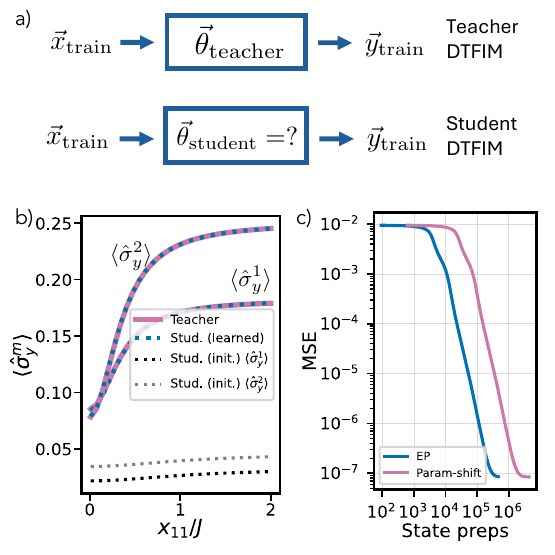}
    \caption{(a) Schematic depicting teacher-student learning, where a `student' system with drives \(\vec \theta_{\rm student}\) is used to learn the input-output relation generated by an identical `teacher' system (with some initially unknown drives $\vec \theta_{\rm teacher}$). (b)  Input-output relations for teacher-student learning with a $N=11$ spin DTFIM model
    (c.f.~Eq.~\eqref{eq:dtfim-ham}) with one input spin and three output spins ($M_{\rm out} = 3$). The input is the spin 11 drive amplitude \(x_{11}\), and the outputs are \(\vec{y} = (\langle \hat \sigma_y^1\rangle, \langle \hat \sigma_y^2\rangle, \langle \hat \sigma_y^3\rangle)\); only two are plotted.  
    Solid pink curves correspond to the target input-output relation defined by the teacher.
    The drive amplitudes \(\vec{\theta}_{\rm student}\) (for spins 1 to 10) are trained either by EP (Sec.~\ref{section:closed-ep-2}) or parameter-shift (Sec.~\ref{section:param-shift}).
    Black dotted `initial' curves correspond to the student output pre-training.  The blue dotted `learned' curves correspond to the outputs of the student system after training has converged. Both methods converge to the same input-output relation.
    We use \(J / \kappa = 5\), and \(\Delta = 0\). We approximate finite differences with \(\beta = 0.01\), and use a learning rate \(\eta = 0.05\). Values for teacher and student drives (before and after training) are included in Table~\ref{table:dtfim-1} in App.~\ref{section:dtfim-teacher-student-examples}.
    (c) The mean squared error (MSE) between the outputs \(\vec y\) and the targets \(\vec y_{\rm train}\) as training progresses. We plot the evolution of the MSE for both training methods, EP and parameter shift, as a function of the number of DTFIM state preparations.}
    \label{fig:teacher-student-compare}
\end{figure}

In Fig.~\ref{fig:teacher-student-compare}, we present an example of teacher-student learning for $N=11, M_{\rm input} =1$ and $M_{\rm out} = 3$ (here $M_{\rm input} \equiv N - M_{\rm train}$). We compare results for training based on EP, versus training based on the more naive parameter-shift method (see Sec.~\ref{section:param-shift}), with all learning hyperparameters fixed across both methods. A further example (corresponding to different random teacher drives) is presented in App.~\ref{section:dtfim-teacher-student-examples}.  
All simulations in this model are performed using a formally exact matrix product state (MPS) representation of the exact steady-state solution (see App.~\ref{section:mps} for details). This allows efficient numerical evaluation of steady-state expectation values limited only by machine precision and does not require MPS simulation of the full master equation.

Fig.~\ref{fig:teacher-student-compare}(c) shows the evolution of the cost function being minimized, the MSE $C \equiv  1/M_{\rm out} \sum_l \| y_l - y^{\rm train}_l \|^2$, as the training progresses for both EP and the parameter shift method applied to our DTFIM.
It is plotted as a function of the total number of distinct DTFIM state preparations.
We clearly see a relative speedup of EP versus parameter shift that scales as expected: the number of state preparations needed in EP for a converged minimal $C$ is smaller by a factor of \(2M_{\rm train}/3 \approx 6.67\) compared to the number needed for parameter-shift.  This is  what would be expected from the general arguments in Sec.~\ref{section:closed-ep} (for a symmetric nudging procedure, which we use).   

It is also interesting to note that the learned student drives (in  Table~\ref{table:dtfim-1} in App.~\ref{section:dtfim-teacher-student-examples}) do not exactly match the teacher drives, despite the input-output relation being learned correctly.
In particular, for this setup where we have a three dimensional function output, the converged student parameters $\theta_j$ only match those of the teacher on the three output spins.  
This is in fact true for all examples tested (see App.~\ref{section:dtfim-teacher-student-examples}), where the student drives match the teacher drives for only the output spins, but not necessarily the internal ones.


\section{EP in a bosonic lattice with quadratic driving}
\label{section:bosonic-lattice}

We now show that our quantum non-equilibrium version of EP can apply for models with more structure and complexity than the DTFIM of the previous section.  The model is a pair driven, dissipative and interacting bosonic lattice system (see Fig.~\ref{fig:bosonic}(a)), where again the NESS is not a simple Gibbs state.  As discussed below, this model now has a non-trivial spatial structure.  Despite this, hTRS still enables a static susceptibility symmetry in the NESS, which in turn makes EP possible and effective.  

\subsection{Model and susceptibility symmetry}
Consider \(N\) bosonic modes (with associated lowering operators \(\{\hat{a}_j\}\)) with parametric (two-photon) drives 
associated with every pair of modes; the complex drive amplitude for a given pair is denoted  \(M_{ij}\).
All drives are at the same frequency, and all modes have the same linear frequency (with a detuning $\Delta$ from the parametric resonance).   
The system also has self-Kerr and cross-Kerr nonlinearities.  We take these all to have magnitude $U/N$, so that the interaction involves the total photon number of the system. Treating the driving within the rotating wave approximation and moving to a rotating frame, we obtain a time-independent Hamiltonian:
\begin{equation}
\label{eq:bosonic-ham}
    \hat{H}_b = \frac{U}{N} \left( \sum_j \hat{n}_j\right)^2 - \Delta \sum_j \hat{n}_j + \sum_{i,j}  \left(M_{ij}\hat{a}_i^{\dag}\hat{a}_j^{\dag} + \text{h.c.}\right)
    \end{equation}
where \(M\) is an arbitrary complex symmetric matrix, and \(\hat n_j\) is a photon number operator, i.e.~\(\hat n_j = \hat a^{\dag}_j \hat a_j\).
Finally, each mode is subject to local single-site photon loss at a dissipation rate \(\kappa\) (taken to be the same for each mode), yielding the Lindbladian
\begin{equation}
\label{eq:bosonic-lind}
    \mathcal{L}_{b}[\hat{\rho}] = -i[\hat{H}_b,\hat{\rho}] + \kappa\sum_j \mathcal{D}[\hat{a}_j] \hat{\rho}.
\end{equation}

To use this system for physical learning, we will again view the drive amplitudes $M_{ij}$ as trainable parameters.  While we have some of the same ingredients as the DTFIM spin model (i.e.~global interaction, driving, local dissipation), there are crucial differences.  In addition to an unbounded local Hilbert space dimension, our system now also has a spatial structure (i.e.~$M_{ij}$ defines a connectivity graph).
In contrast, the DTFIM is essentially zero-dimensional.
These differences make the physics of this model far richer, and observables are generally more sensitive to variations of drive amplitudes (e.g.~compare  Fig.~\ref{fig:bosonic}(b) against Fig.~\ref{fig:teacher-student-compare}(b) for the DTFIM).
Also, note that the number of tuneable parameters \(M_{ij}\) now grow {\it quadratically} in the system size \(N\), whereas it only grew linearly in the DTFIM.  

The basic physics of this setup involves a competition between driving, dissipation and interactions.  
The drives populate the modes with photon pairs and generate squeezing correlations between them.  
This is opposed by the loss, detuning and the Kerr interactions.
The competition between these two effects gives rise to dissipative phase transitions, non-trivial steady states and a host of rich quantum phenomena~\cite{roberts2023competition}.
A number of related models have also been studied numerically in the literature~\cite{savona2017spontaneous, lebreuilly2019stabilizing,rota2019quantum}.
Surprisingly, for all choices of parameters, the pair-driven bosonic model described by Eq.~\eqref{eq:bosonic-ham} and \eqref{eq:bosonic-lind} has hTRS and is exactly solvable, something demonstrated in Ref.~\cite{roberts2023competition}.  We make use of this structure in what follows.

We now move to the key new result of this section: the existence of a susceptibility symmetry in the NESS that enables EP.
As shown in App.~\ref{section:bosonic-proof} and sketched below, hTRS again allows us to derive non-trivial constraints on static susceptibilities in the NESS, for any choice of model parameters.  However, for these to be {\it useful} for EP, we need an additional constraint beyond simply having hTRS.  We also need to restrict our model such that (in some fixed local gauge) all drive amplitudes $M_{ij}$ are real. 
This is markedly different from the case of the DTFIM, where a useful susceptibility symmetry exists even if drive amplitudes are allowed to be complex (see App.~\ref{section:dtfim-reciprocity-derivation} for details).
This additional constraint in our bosonic model is suggestive of a conventional time-reversal symmetry, i.e.~there are no non-trivial Aharonov-Bohm fluxes associated with the connectivity between modes encoded in the $M$ matrix.   

Similar to the DTFIM, the observables that turn out to be relevant for EP are associated with coherence in the NESS.  We introduce for each pair of modes $i,j$ the Hermitian operators
  \begin{align}
    \label{eq:xij}\hat{X}_{ij} & \equiv \frac{1}{2}(\hat{a}_i \hat{a}_j + \hat{a}^{\dag}_i \hat{a}^{\dag}_j) \\  \label{eq:yij}\hat{Y}_{ij} & \equiv \frac{1}{2i}(\hat{a}_i \hat{a}_j- \hat{a}^{\dag}_i \hat{a}^{\dag}_j)  \\ \hat{N}_{ij} & \equiv \hat{n}_i + \hat{n}_j.
\end{align}
$\hat{X}_{ij}$ and $\hat{Y}_{ij}$ characterize coherent pair correlations between the modes $i$ and $j$, whereas $\hat{N}_{ij}$ is the total photon number in these modes.  

As derived in detail in App.~\ref{section:bosonic-proof}, if we restrict our model to always having real-valued drive amplitudes $M_{ij}$, hTRS implies a non-trivial susceptibility symmetry associated with the drive-dependent NESS:
\begin{equation}
\label{eq:bosonic-recRelation}
   \frac{1}{(2 - \delta_{kl})} \frac{\partial \langle \hat{Y}_{ij} \rangle}{\partial M_{kl}} =   \frac{1}{(2 - \delta_{ij})}\frac{\partial \langle \hat{Y}_{kl} \rangle}{\partial M_{ij}}
\end{equation}
where \(\delta_{ij}\) is the Kronecker delta.
Here, as before, expectation values are taken in the NESS (which is itself a function of the real-valued $M_{ij}$).
Note that here $M_{ij}$ and $M_{ji}$ denote the same parameter since our model is only defined for a symmetric $M$ matrix.
Therefore, the weights $(2-\delta_{ij})$ account for the fact that our Hamiltonian in Eq.~\eqref{eq:bosonic-ham} was defined to sum over all $(i,j)$, and hence double counts the contribution associated with modes for which $i \neq j$.  

Note that there is a structural similarity to the corresponding relation for the DTFIM, Eq.~\eqref{eq:tfim-reciprocity}.
In the DTFIM, the relevant observables were the coherences $\hat{Y}_j$, which did not explicitly appear in the Lindbladian, and which were conjugate to the drive operators $\hat{X}_j$.
Similarly, in the bosonic model the $\hat{Y}_{ij}$ do not appear directly in the Lindbladian; instead the real drive amplitudes $M_{ij}$ couple directly to the conjugate quantities $\hat{X}_{ij}$.    

As in the case of the DTFIM, the proof of our bosonic susceptibility symmetry relies crucially on hTRS, and at a heuristic level directly connects to Fig.~\ref{fig:thms}.  One starts with a trivial kind of gauge covariance under a $U(N)$ transformation that mixes modes $\hat{a}_j$.  This allows one to trivially relate the NESS for one choice of $M$ to another $M'$ if the two matrices are unitarily related.  
hTRS turns this into something much more powerful: we can now relate the NESS for {\it any} choice of $M$ matrix to any other choice $M'$.
This is the formal  content of our Theorem~\ref{theorem:bosonic} (see App.~\ref{section:bosonic-proof}).  This non-trivial relation between different NESS then lets us directly prove Eq.~\eqref{eq:bosonic-recRelation}.

\begin{figure}
    \centering
    \includegraphics[width=\linewidth]{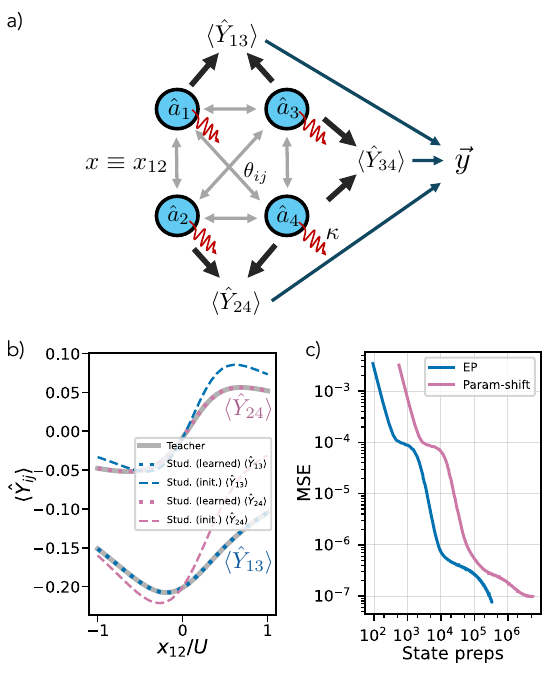}
    \caption{(a) Schematic depicting 
    the encoding of an input-output relation in a dissipative, pair-driven bosonic lattice (here with $N=4$ modes).  We show a setup with one input parameter and three output observables. The  single input \(x \equiv x_{12}\) and training parameters \(\theta_{ij}\) (gray arrows) are driving amplitudes (matrix elements of \(M\)).
    The three outputs are the observables \(\vec{y} \equiv (\langle \hat{Y}_{13}\rangle, \langle \hat{Y}_{34}\rangle,\langle \hat{Y}_{24}\rangle)\). There are 10 tuneable couplings in the system, and all but the input are trainable parameters, \(\theta_{ij}\). (b) Input-output relations for teacher-student training in the four mode bosonic model depicted in (a). We plot two of the three outputs. Solid gray curves correspond to the target relation set by the teacher. Dashed `initial' curves correspond to student outputs pre-training; dotted `learned' curves correspond to these outputs post-training. Both EP and parameter-shift are used, and converge to the same relation. We use system parameters \(U/\kappa = 1\), \(\Delta = 0\). 
    We approximate finite differences with \(\beta = 0.01\), and use a learning rate \(\eta = 0.05\). Values for teacher and student drives are included in Table~\ref{table:bosonic} in App.~\ref{section:dtfim-teacher-student-examples}. (c) Mean square error (MSE) as training progresses, for both EP and parameter-shift, as a function of the number of bosonic lattice steady state preparations.}
    \label{fig:bosonic}
\end{figure}

\subsection{EP-based training using NESS}
\label{section:bosonic-ep}
In this section we describe how the NESS of the bosonic system can be used to perform EP-based training.
 The key difference from the training setup for the DTFIM described in Sec.~\ref{section:dtfim} is that we encode inputs and training parameters in a driving matrix \(M\) that couples pairs of modes, as opposed to a vector of local single site drives \(\vec{\Omega}\).
For a system of \(N\) modes, the bosonic model has \(N(N+1)/2\) tunable real parameters, while the DTFIM only had \(N\) parameters for \(N\) spins.

It is useful to define sets \(A_{\rm input}\), \(A_{\rm out}\), \(A_{\rm int}\) which contain tuples \(\{(i, j)\}\) corresponding (respectively) to the input, output and additional internal training drives.
To be explicit, we partition the elements of our drive matrix as
\begin{equation}
    M_{ij} \equiv 
     \begin{cases}
    \theta'_{ij} (\delta_{ij} + 1)/2 & (i,j) \in A_{\rm out}\\
    \theta''_{ij} (\delta_{ij} + 1)/2 & (i,j) \in A_{\rm int} \\
    x_{ij} (\delta_{ij} + 1)/2 & (i,j) \in  A_{\rm input}.
    \end{cases}
\end{equation}
The inputs are encoded in \(x_{ij}\), while the \(\theta_{ij}\) (which refers to all couplings \(\theta'_{ij}\) and \(\theta''_{ij}\)) are the trainable parameters.
Note that here by definition we have that $(i,j)$ is the same pair as $(j,i)$, and we therefore have one parameter $x_{ij}$ and $\theta_{ij}$ for each pair.
The rescaling factor $(\delta_{ij} + 1)/2$ cancels the weights in Eq.~\eqref{eq:bosonic-recRelation}, and the susceptibility symmetry takes the more natural form:
\begin{equation}
     \frac{\partial \langle \hat{Y}_{ij} \rangle}{\partial \theta_{kl}} = \frac{\partial \langle \hat{Y}_{kl} \rangle}{\partial \theta_{ij}}.
\end{equation}
This is the form required for EP (see App.~\ref{section:EP-proof}), and is why the $\theta_{ij}$ as opposed to $M_{ij}$ are the natural training parameters.

Having specified how the drive elements in our model will be used either as inputs to our learning problem or as internal trainable parameters, we now turn to identifying the outputs.  These are defined as
\begin{equation}
    y_{ij} = \langle \hat{Y}_{ij}\rangle, \quad (i,j) \in A_{\rm out} .
\end{equation}
We again emphasize that only one output $y_{ij}$ is defined for a given pair of modes $(i,j)$ or equivalently $(j,i)$.
Training proceeds similarly as in the DTFIM; the gradient update is given by
\begin{equation}
    \frac{\partial C}{\partial \theta_{ij}} \approx \frac{\Tr\left(\hat{Y}_{ij} \hat{\rho}_{\beta}\right) - \Tr\left(\hat{Y}_{ij} \hat{\rho}_{0}\right)}{\beta}.
\end{equation}
Here, the steady state \(\hat \rho_0\) (\(\hat \rho_\beta\)) in the free (nudged) phase is defined identically as in Eq.~\eqref{eq:ness-free} (Eq.~\eqref{eq:ness-nudged}) where now \(\hat{\rho}_{\rm ss}\) is the steady state of the Lindbladian Eq.~\eqref{eq:bosonic-lind}.

In Fig.~\ref{fig:bosonic}(b-c), we present results for teacher-student training (see Sec.~\ref{section:dtfim} for details on the training setup) on an \(N = 4\) mode system used to map a one dimensional input to a three dimensional output. 
To compute steady-state expectation values we constructed a formally exact matrix product state representation of the exact steady-state solution (see App.~\ref{section:bosonic-mps} for details). This MPS construction is very similar to the DTFIM with one significant difference: the local Hilbert space dimension of the DTFIM is 2, but it is unbounded for the bosonic model.
While this requires us to introduce a local photon cutoff for the bosonic MPS, it also adds significantly more complexity to the state. We note that the local-photon-cutoff approximation can be controlled by increasing the photon cutoff.

Once again, we compare training using EP and parameter-shift (with all learning hyperparameters fixed across both training methods), and use a symmetric nudging procedure to compute gradients. 
One again sees (as expected) a significant learning speedup of EP over parameter shift.  To be explicit, the number of state preparations needed for convergence of EP is improved by a factor of \(2(|A_{\rm out}| + |A_{\rm int}|)/3 = 6\) where \(|.|\) denotes the cardinality of the set.
What is more interesting in this example is the structure of the input-output relation in Fig.~\ref{fig:bosonic}(b): the outputs \(y_{ij}\) are distinct, non-monotonic functions of the input (in contrast to the DTFIM outputs in Fig.~\ref{fig:teacher-student-compare}(b)).
One may then expect that this input-output relation could represent a much richer class of functions than the DTFIM even at smaller system sizes, potentially providing higher expressivity.
At a heuristic level, this can be attributed to the additional complexity of the bosonic model that arises from it having an effective spatial structure.  
While a full analysis of the expressive power of this system is outside of the scope of the current work, it would be an interesting avenue for future exploration.

\section{Outlook}
\label{section:discussion}
In this work, we established that the highly effective equilibrium propagation strategy for physical learning can be extended to non-equilibrium, driven-dissipative quantum systems.  This is possible despite the absence of any meaningful free energy in the non-equilibrium setting.  We showed that the crucial ingredient was hidden time-reversal symmetry (a kind of quantum detailed balance).  
We derived fundamentally new general relations on how the NESS flows under a change of parameters using hTRS (Theorems~\ref{theorem:tfim} and~\ref{theorem:bosonic}) in two representative many-body models:  a dissipative Ising model (Sec.~\ref{section:dtfim}) and a driven-dissipative bosonic lattice (Sec.~\ref{section:bosonic-lattice}), and showed how these directly enable EP.
We benchmarked training in both models using a teacher-student learning method, confirming that EP provides a substantial scaling advantage over the direct parameter-shift method in these non-equilibrium settings.

Having established that EP is possible in driven-dissipative quantum systems, a natural next question is that of expressivity of non-equilibrium quantum systems with hTRS.
The structure of the input-output relation of the bosonic model (see Fig.~\ref{fig:bosonic}(b)) is already encouraging.
Exploring more systematic metrics to benchmark the kinds of functions these models can learn is an interesting avenue for future work.
Furthermore, we need not limit ourselves to just these two models:
there are a number of models that are known to have hTRS (e.g. driven spin chains \cite{lingenfelter2024exact, yao2025hidden}, interacting fermion models~\cite{lingenfelter2026exact}).  It would be interesting to explore our EP-based learning approach with these systems.  

Another extremely promising direction for future work is to explore learning in systems that only approximately have hTRS, and thus only approximately satisfy our susceptibility symmetries (Eq.~\eqref{eq:tfim-reciprocity}, Eq.~\eqref{eq:bosonic-recRelation}).
This symmetry is used to approximate gradients needed to descend a loss landscape (gradient descent).
Crucially, the gradients do not have to be exact: approximate gradients will still move training in the right direction, towards a minimum of the loss landscape, albeit training might take longer. This is the spirit of much related work in the field (see e.g.~\cite{cin2025training}).
Using models that only approximately have hTRS could vastly expand the class of usable systems.  This would both avoid issues of experimental fine-tuning, but could also potentially be used to enhance expressivity.

\section*{ACKNOWLEDGMENTS}
We thank Arvind Murugan for helpful discussions. Research was sponsored by the Army Research Office and was accomplished under Grant Number W911NF-25-1-0261.  AC also acknowledges support from the Simons Foundation through a Simons Investigator Award (Grant No. 669487).

\textit{AI statement.} A large language model (Claude, Anthropic) was used to write and debug code for numerical simulations. It helped refine and generalize the key proofs of Theorems~\ref{theorem:tfim} and~\ref{theorem:bosonic}. All calculations presented were derived and verified independently by the authors. All text was written by the authors, although Claude helped catch typos and notational inconsistencies. The authors take full responsibility for the content of the work.

\onecolumngrid
\begin{appendix}
\section{Derivation of EP gradient update}
\label{section:EP-proof}

We use the susceptibility symmetry Eq.~\eqref{eq:reciprocity-closedquantum} to derive Eq.~\eqref{eq:grad-ep-closed} (EP).
Starting from Eq.~\eqref{eq:cost-func-dotproduct}, we have
\begin{align}
\label{eq:rec-derivation-1}
    \frac{\partial C}{\partial \theta_j}  & = \sum_l \varepsilon_l \frac{\partial \langle \hat{O}_l\rangle}{\partial \theta_j} \\
    \label{eq:rec-derivation-2}
& = \sum_l \varepsilon_l \frac{\partial \langle \hat{O}_j\rangle}{\partial \theta'_l} 
\end{align}
where in the second line we have applied the susceptibility symmetry.
Now, we define a new ground state that forms after we perturb the drive \(\theta_l' \to \theta_l' + \beta \varepsilon_l\), as
\begin{equation}
    \ket{\Psi_l} \equiv \ket{\Psi(\vec{x}, \vec{\theta} + \beta \varepsilon_l \vec{n}_l)}.
\end{equation}
Nudging the drive proportional to the error vector \(\varepsilon\) is crucial to remove the need to do the dot product in Eq.~\eqref{eq:rec-derivation-1}.
We can then approximate the derivative in Eq.~\eqref{eq:rec-derivation-2} as a finite-difference between the initial and new ground state as:
\begin{align}
\frac{\partial C}{\partial \theta_j}& \approx \sum_l \varepsilon_l \frac{\langle \Psi_l| \hat{O}_j | \Psi_l \rangle - \langle \Psi_0 | \hat{O}_j |\Psi_0 \rangle}{\beta \varepsilon_l}\\
\label{eq:rec-derivation-4}
& = \sum_l \frac{\langle \Psi_l| \hat{O}_j | \Psi_l \rangle - \langle \Psi_0 | \hat{O}_j |\Psi_0 \rangle}{\beta}.
\end{align}
We're nearly there! We now only need to express Eq.~\eqref{eq:rec-derivation-4} in terms of the state \(\ket{\Psi_0'}\) as defined in Eq.~\eqref{eq:ep-nudged-close}.
This state is formed by perturbing all drives \(\vec{\theta}' \to \vec{\theta}' + \beta \vec{\varepsilon}\).
It will be useful to compare the expansion of \(\langle \hat{O}_j \rangle\) in powers of \(\beta\), for both states \(\ket{\Psi_0'}\) and \(\ket{\Psi_l}\). Explicitly, we have 
\begin{align}
    \langle \Psi_0'|\hat{O}_j | \Psi_0' \rangle & = \langle \Psi_0|\hat{O}_j | \Psi_0 \rangle + \beta \sum_l \varepsilon_l \frac{\partial \langle \Psi_0|\hat{O}_j | \Psi_0 \rangle}{\partial \theta_l} + \mathcal{O}(\beta^2) \\ 
    \langle \Psi_l|\hat{O}_j | \Psi_l \rangle & = \langle \Psi_0|\hat{O}_j | \Psi_0 \rangle + \beta \varepsilon_l \frac{\partial \langle \Psi_0|\hat{O}_j | \Psi_0 \rangle}{\partial \theta_l} + \mathcal{O}(\beta^2).
\end{align}
Then, to linear order we have 
\begin{equation}
\label{eq:rec-derivation-5}
    \langle \Psi_0'|\hat{O}_j | \Psi_0' \rangle  - \langle \Psi_0|\hat{O}_j | \Psi_0 \rangle \approx \sum_l \left(\langle \Psi_l| \hat{O}_j | \Psi_l \rangle - \langle \Psi_0 | \hat{O}_j |\Psi_0 \rangle\right).
\end{equation}
This approximation is valid since we are operating in the regime of linear response by approximating derivatives using finite differences.
Applying Eq.~\eqref{eq:rec-derivation-5} to Eq.~\eqref{eq:rec-derivation-4}, we recover exactly Eq.~\eqref{eq:grad-ep-closed}; we are done.
The key insight here is that the susceptibility symmetry removes the need to perform a dot product, since the observable in Eq.~\eqref{eq:rec-derivation-2} becomes independent of the sum over \(l\).

\section{Hidden time-reversal symmetry and constraints on the purification of the steady state}
\label{appendix:htrs}

A useful way to understand constraints set by hTRS is to re-express the defining conditions in Eq.~\eqref{eq:htrs-heff} and Eq.~\eqref{eq:htrs-ck} as constraints on a purification of the steady state density matrix \cite{roberts2021hidden,noachtar_exact_2026}.  This formulation will greatly simplify the proof of the crucial static susceptibility symmetries needed for EP-based learning.  
In the main text, hTRS was defined in terms of an anti-linear operator $\hat \Psi$ which acted as a square root of the steady state density matrix, $\hat \rho_{\rm ss} = \hat \Psi \hat \Psi^\dagger$.  By using a standard Choi-Jamiolkowski mapping, we can convert (i.e.~vectorize) the operator $\hat \Psi$ to a state vector $\ket \Psi$ in a doubled Hilbert space.  Calling the two copies of the Hilbert space A and B, it immediately follows that this state is a purification of the steady state: 
\begin{align}
\label{eq:pure-ss}
        \hat \rho_{\rm ss} = \Tr_B{\ket{\Psi}\bra{\Psi}}.
    \end{align}
Defining system-A and system-B operators in the obvious manner (i.e.~\(\hat O_A \equiv \hat O \otimes  \mathbb{I}_B\),  \(\hat O_B \equiv  \mathbb{I}_A \otimes \hat O\)), Eqs.~(\ref{eq:htrs-ck}),(\ref{eq:htrs-heff})
become constraint conditions on the purification:
\begin{align}
\label{eq:darkstate-1}
\forall i  \,\,\, (\hat c_{i, A} - \hat c_{i, B}) \ket{\Psi} &= 0,  \\
\label{eq:darkstate-2}
(\hat H_{\text{eff}, A} - \hat H_{\text{eff}, B}) \ket{\Psi} &= 0.
\end{align}
Here \(\hat H_{\text{eff},\alpha}\) and \(\hat c_{j, \alpha}\) represent the effective Hamiltonian and jump operators respectively on system \(\alpha\), with \(\alpha \in [A,B]\).
It also follows that {\it any} pure state in the doubled Hilbert space satisfying Eqs.~\eqref{eq:darkstate-1} and \eqref{eq:darkstate-2} necessarily is a purification of the steady state; hence, these equations uniquely determine the steady state.  

We next prove a useful lemma regarding
certain bilinear expectation values in the purification.

\begin{lemma}
\label{lemma:z}
Consider a Lindbladian with jump operators \(\{\hat c_k\}\) that has hTRS, and consequently a purification of the steady state \(\ket{\Psi}\) satisfying Eq.~\eqref{eq:darkstate-1},\eqref{eq:darkstate-2}.
Then, we have
\begin{equation}
\langle \Psi | \hat c^{\dag}_{i,A} \hat c_{j,A} |\Psi \rangle  = \langle \Psi | \hat c^{\dag}_{i,B} \hat c_{j,B} |\Psi \rangle  \quad \forall i,j.
\end{equation}
\end{lemma}
\begin{proof}
We can rewrite
\begin{equation}
\langle \Psi | \hat c^{\dag}_{i,A} \hat c_{j,A} |\Psi \rangle - \langle \Psi | \hat c^{\dag}_{i,B} \hat c_{j,B} |\Psi \rangle = \langle \Psi | \hat c^{\dag}_{i,A} \left(\hat c_{j,A} - \hat c_{j,B} \right) | \Psi \rangle +  \langle \Psi | \left( \hat c^{\dag}_{i,A} -  \hat c^{\dag}_{i,B}\right) \hat c_{j,B} | \Psi \rangle.
\end{equation}
The first term on the RHS vanishes from Eq.~\eqref{eq:darkstate-1}, and the second term vanishes from its conjugate \(\bra{\Psi} (\hat c^{\dag}_{i,A} - \hat c^{\dag}_{i,B}) = 0\). 
This proves the lemma.
\end{proof}

\section{Proof of susceptibility symmetry in DTFIM}
\label{section:dtfim-reciprocity-derivation}
\label{tfim:derivation}

We now prove the susceptibility symmetry Eq.~\eqref{eq:tfim-reciprocity} which enables EP in the DTFIM.
We first state and prove Theorem~\ref{theorem:tfim}, which gives a simple way to calculate changes in the DTFIM steady state due to changes in drive amplitude and phases.  
The theorem shows that hTRS allows one to non-trivially extend a simple gauge covariance of our Lindbladian to a much more general relation.  
After proving the theorem in App.~\ref{app:tfim-htrs} we use it in App.~\ref{app:dtfim-proof} to prove Eq.~\eqref{eq:tfim-reciprocity}.

In what follows, we consider a slightly generalized version of the DTFIM model in the main text (c.f.~Eq.~\eqref{eq:dtfim-ham}), where we now allow the drive amplitudes to be complex.  The effective Hamiltonian and jump operators are
\begin{align}
\label{eq:dtfim-heff}
    \hat H_{\rm eff}(\vec\Omega)&=\frac{J}{4N}\sum_{i\neq j}\hat\sigma^i_z\hat\sigma^j_z
+\frac12\sum_j\left(\Delta\hat\sigma^j_z+\Omega_j\hat\sigma^j_+ +\Omega^*_j\hat\sigma^j_-\right)
-\frac{i\kappa}{2}\sum_j\hat\sigma^j_+\hat\sigma^j_-,
\\ \hat c_j&=\sqrt\kappa\,\hat\sigma^j_- .
\label{eq:dtfim-ck}
\end{align}
Throughout the derivation, we will work with the formulation of hTRS described in App.~\ref{appendix:htrs}, which is based on a purification $\ket \Psi$ of the steady state and system operators in the corresponding doubled Hilbert space.
Since the DTFIM has hTRS, this purification satisfies the conditions 
Eq.~\eqref{eq:darkstate-1} and Eq.~\eqref{eq:darkstate-2} for any choice of system parameters.  We explicitly index the purification corresponding to a particular set of drives $\vec \Omega$ as \(  \ket{\Psi (\vec \Omega)}\).
Note that we do not assume these states to be normalized.

To relate steady state purifications for different drive parameters to one another, it will be useful to introduce the operator  \(\hat{Z}_i \equiv \frac{1}{2}\left(\hat \sigma^i_{z, A}  + \hat \sigma^i_{z, B}\right)\)
which represents a total charge in the doubled system.  Exponentiating this operator gives a symmetry operator that transforms lowering and raising operators \(\hat \sigma^i_{\pm,\alpha}\) (here \(\alpha \in [A,B]\)) in a simple manner: these operators are charged with respect to the $\hat{Z}_i$.  
Using the standard commutation relations \([\hat \sigma_z^i, \hat \sigma_{\pm}^{j}] = \pm 2\delta_{ij} \hat \sigma_{\pm}^{i}\)
and Hausdorff's lemma, we have for any complex $\lambda$: 
\begin{equation}
\label{eq:sigz-pm-conjugation}
e^{-\lambda \hat Z_i} \hat{\sigma}_{\pm,\alpha}^{j} e^{\lambda \hat Z_i} =  e^{\mp \delta_{ij}\lambda} \hat \sigma_{\pm,\alpha}^{j}.
\end{equation}
For $\lambda$ purely imaginary, this is a simple $U(1)$ gauge transformation, while for general $\lambda$ there is also a rescaling of the operators.  
This basic property combined with hTRS will let us prove the central theorem that our proof relies on.  We first state this theorem.  

\begin{theorem}
\label{theorem:tfim}
Consider a DTFIM defined by Eq.~\eqref{eq:dtfim-heff} and Eq.~\eqref{eq:dtfim-ck} with some fixed set of drive amplitudes \(\vec{\Omega}\).
Let \(\ket{\Psi (\vec{\Omega})}\) be a purification (not necessarily normalized) of its steady state that satisfies the hTRS conditions in Eq.~\eqref{eq:darkstate-1} and Eq.~\eqref{eq:darkstate-2}.
Then, a second DTFIM with drives \(\Omega'_j = e^{\lambda_j}\Omega_j\) for any parameters \(\lambda_j \in \mathbb C\) (with all other system parameters identical to the first), has a purification (also not normalized) of its steady state \(\ket{\Psi (\vec{\Omega}')}\) given by \(\ket{\Psi (\vec{\Omega}')} = e^{\sum_j \lambda_j \hat Z_j}\ket{\Psi (\vec{\Omega})}\).  This purification
also satisfies the hTRS conditions.
\end{theorem}

At a heuristic level, Theorem 1 gives us a simple and direct way to calculate how the dissipative steady state changes as drive amplitudes change.  This in turn will let us strongly constrain how steady-state expectation values change as drive parameters vary.  
Note that in the case where the $\lambda_j$ are all imaginary, Theorem~\ref{theorem:tfim} becomes trivial: it just reflects covariance of the original Lindbladian and steady state under $U(1)$ gauge transformations.  Stated explicitly, if we transform the Lindbladian using a unitary $\hat{U} = \exp \left(i/2 \sum_j \varphi_j \hat{\sigma}^j_z \right)$, the steady state transforms the same way.  This unitary simply maps the DTFIM Lindbladian to a new DTFIM Lindbladian with shifted drive phases.  This basic gauge covariance of course does not require hTRS. Theorem 1 is much more powerful: it relates steady states for two sets of drives $\vec \Omega$, $\vec \Omega'$ that differ both in phase {\it and in magnitude}.  As we will see, this crucially requires hTRS.

\subsection{Proof of Theorem~\ref{theorem:tfim} using hTRS}
\label{app:tfim-htrs}
Let \(\ket{\Psi(\vec\Omega)}\) be the purification of a DTFIM steady state with drives \(\vec \Omega\); it thus necessarily satisfies hTRS conditions Eq.~\eqref{eq:darkstate-1} and Eq.~\eqref{eq:darkstate-2}. Let us also define \(\ket{\psi} \equiv e^{\sum_j \lambda_j \hat Z_j}\ket{\Psi(\vec \Omega)}\). We will show that for arbitrary complex \(\lambda_j \in \mathbb{C}\), \(\ket{\psi}\) satisfies the same hTRS conditions for some modified set of drives \(\vec \Omega'\), and is consequently the purification of the DTFIM steady state with these drives, i.e.~\(\ket{\psi} = \ket{\Psi(\vec \Omega')}\).

Using Eq.~\eqref{eq:sigz-pm-conjugation}, we have for some \(\vec{\lambda} \in \mathbb C^N\):
\begin{equation}
e^{-\sum_j \lambda_j \hat Z_j} \left(\hat \sigma_{-,A}^i -  \hat \sigma_{-,B}^i\right)e^{\sum_j \lambda_j \hat Z_j} 
\ket{\Psi (\vec \Omega)} = e^{\lambda_i} \underbrace{\left(\hat \sigma_{-,A}^i -  \hat \sigma_{-,B}^i\right) \ket{\Psi (\vec \Omega)}}_{= 0}.
\end{equation}
The RHS vanishes via Eq.~\eqref{eq:darkstate-1}.
We then multiply both sides on the left by \(e^{\sum_j \lambda_j \hat{Z}_j}\):
\begin{equation}
\label{eq:darkstate-1-psi}
\left(\hat \sigma_{-,A}^i -  \hat \sigma_{-,B}^i\right)e^{\sum_j \lambda_j \hat Z_j} \ket{\Psi (\vec \Omega)} =  0 \implies  \left(\hat \sigma_{-,A}^i -  \hat \sigma_{-,B}^i\right) \ket{\psi} = 0.
\end{equation}
We thus have that \(\ket{\psi}\) satisfies the jump operator hTRS condition in Eq.~\eqref{eq:darkstate-1}.

To show that the second, drive-dependent hTRS condition is also satisfied by \(\ket{\psi}\), we first consider how the single-system effective Hamiltonian (for a set of drives $\Omega'_j = e^{\lambda_j} \Omega_j$)  transforms under our generalized symmetry operator:
\begin{align}
e^{-\sum_j \lambda_j \hat \sigma_z^j/2} \hat H_{\text{eff}}(\vec{\Omega}')e^{\sum_j \lambda_j \hat \sigma_z^j/2} & = \frac{J}{4N} \sum_{i\neq j} \hat{\sigma}_z^i \hat{\sigma}_z^j + \frac{1}{2} \sum_{j=1}^N \left(
        \Delta \hat{\sigma}_z^j  - i \kappa \hat \sigma_+^j \hat \sigma_-^j \right) + \frac{1}{2}\sum_j \left(e^{- \lambda_j} \Omega'_j \hat \sigma_+^j + e^{\lambda_j} {\Omega'_j}^{*}\hat\sigma_-^j \right) \\
        \label{eq:heff-transform}
        & = \hat H_{\text{eff}} (\vec{\Omega}) + \sum_j \frac{\Omega^*_j}{2} \left(e^{2 {\rm Re}\lambda_j} - 1 \right) \hat \sigma_-^j.
\end{align}
The $Z$-only terms are left unchanged, whereas the drive terms are modified.  The transformation scales the coefficients of the raising operator drive terms correctly (changing $\vec \Omega'$ to $\vec \Omega$), but if $\textrm{Re } \lambda_j \neq 0$, it modifies the lowering operator terms incorrectly (last term in Eq.~\eqref{eq:heff-transform}).   

We next note that on the doubled system, Eq.~\eqref{eq:heff-transform} implies:
\begin{align}
e^{-\sum_j \lambda_j \hat Z_j} \left(\hat H_{\text{eff}, A} (\vec{\Omega}') - \hat H_{\text{eff}, B} (\vec{\Omega}')\right) e^{\sum_j \lambda_j \hat Z_j} = & \hat H_{\text{eff}, A} (\vec{\Omega}) - \hat H_{\text{eff}, B} (\vec{\Omega}) \nonumber \\ &+ \sum_j \frac{\Omega^*_j}{2} \left(e^{2 {\rm Re}\lambda_j} - 1 \right) \left(\hat \sigma_{-,A}^j -  \hat \sigma_{-,B}^j\right).
\label{eq:heff-transform-double}
\end{align}
We still have an ``error" term on the RHS of Eq.~\eqref{eq:heff-transform-double}, but note that it involves differences of system A and B lowering operators.  While this term is a non-zero operator, the hTRS conditions Eq.~\eqref{eq:darkstate-1} and Eq.~\eqref{eq:darkstate-2} imply that the RHS of Eq.~\eqref{eq:heff-transform-double} vanishes when acting on the state \(\ket{\Psi (\vec\Omega)}\). 
It follows that the LHS must also annihilate this state:
\begin{equation}
e^{-\sum_j \lambda_j \hat Z_j} \left(\hat H_{\text{eff}, A} (\vec{\Omega}') - \hat H_{\text{eff}, B} (\vec{\Omega}')\right) \underbrace{e^{\sum_j \lambda_j \hat Z_j} \ket{\Psi (\vec\Omega)}}_{= \ket{\psi}} = 0.
\end{equation}
We can now multiply both sides on the left by \(e^{\sum_j \lambda_j \hat Z_j}\), yielding
\begin{equation}
\label{eq:darkstate2-psi}
    \left(\hat H_{\text{eff}, A} (\vec{\Omega}') - \hat H_{\text{eff}, B} (\vec{\Omega}')\right) \ket{\psi} = 0.
\end{equation}

Eqs.~\eqref{eq:darkstate-1-psi} and~\eqref{eq:darkstate2-psi} show that the state \(\ket{\psi}\)  satisfies the hTRS conditions Eqs.~\eqref{eq:darkstate-1} and~\eqref{eq:darkstate-2} for the drive parameters \(\vec \Omega'\).
It follows that \(\ket \psi\) must necessarily be a purification of the steady state of a DTFIM with drives \(\vec \Omega'\), i.e.~\(\ket \psi = \ket{\Psi(\vec \Omega')}\).
Thus, we have proved Theorem~\ref{theorem:tfim}.

\subsection{Proof of susceptibility symmetry in the DTFIM}
\label{app:dtfim-proof}

We now use Theorem 1 (which describes the flow of the steady state purification $\ket{\Psi(\vec \Omega)}$ under changes of the drive amplitudes $\vec \Omega$) to prove the crucial susceptibility symmetry of the DTFIM.
For concreteness, we fix all DTFIM parameters except the drive amplitudes, and consider how the steady state varies with changes in $\vec \Omega$.  
The corresponding family of DTFIM steady state purifications are denoted \(\ket{\Psi (\vec \Omega)}\).  We will use Theorem~\ref{theorem:tfim} to parameterize these in terms of a single reference state \(\ket{\Psi (\vec \Omega_0)}\); for simplicity, we take this reference set of drives to be non-zero and uniform, $\vec \Omega_0 = (\bar{\Omega},\bar {\Omega},... )$.
We then parameterize $\Omega$ in terms of a vector 
$\vec \lambda$:
\begin{equation}
\label{eq:lambda-tfim}
    \lambda_j \equiv  \log  (\Omega_j/ \bar \Omega)
\end{equation}
Our family of steady state purifications can then be written as:
\begin{align}
  \ket{ \Psi(\vec \Omega)} & =   e^{\sum_j  \lambda_j \hat Z_j} \ket{\Psi (\vec \Omega_0)}.
\end{align}

The above parameterization fixes the normalization of our family of states \(\mathcal N(\vec \Omega) = \braket{\Psi (\vec \Omega)}{\Psi ( \vec \Omega)}\)
in terms of the reference state:
\begin{align}
\label{eq:norm-tfim}
    \mathcal N( \vec \Omega) & = \bra{ \Psi (\vec \Omega_0)} e^{2\sum_j ({\rm Re} \lambda_j) \hat Z_j} \ket{\Psi ( \vec \Omega_0)}.
\end{align}
Differentiating this equation for state normalization with respect to $\textrm{Re } \lambda_i$ then yields:
\begin{align}
\label{eq:norm-derivative-tfim}
    \frac{\partial \mathcal N( \vec \Omega)}{\partial({\rm Re}\lambda_i)} &  = 
    2 \bra{\Psi (\vec \Omega)} \hat Z_i \ket{\Psi (\vec \Omega)}.
\end{align}
Next, recall from Lemma~\ref{lemma:z} that the RHS here can be written in terms of a single system expectation value: we have \(\bra{\Psi (\vec \Omega)} \hat Z_i \ket{\Psi (\vec \Omega)} = \bra{\Psi (\vec \Omega)} \hat \sigma_{z,A}^i \ket{\Psi (\vec \Omega)}\).
To express this as a normalized expectation value $\langle \hat \sigma_z^i \rangle$, we divide both sides by $\mathcal{N}$, yielding:
\begin{equation}
\label{eq:logN}
\frac{\partial \log \mathcal N(\vec \Omega)}{\partial ({\rm Re}\lambda_i)} = 2\langle \hat \sigma_z^i \rangle.
\end{equation}
We have thus established that the normalization function $\mathcal{N}(\vec \Omega)$ is a generating function for $\hat \sigma_z$ expectation values in the steady state.

The next step is to connect the Eq.~\eqref{eq:logN} to the observables that appear in the relevant static susceptibility.  
Writing the complex drives \(\Omega_j = \omega_je^{i \varphi_j}\) for \(\omega_j \in \mathbb R\), the Heisenberg equation (obtained from the DTFIM Lindbladian) for \(\hat \sigma_z\) is:
\begin{equation}
    \partial_t \langle \hat \sigma_z^i \rangle =  \omega_i\left(\cos \varphi_i \langle \hat \sigma_y^i \rangle + \sin \varphi_i \langle \hat \sigma_x^i \rangle\right) - \kappa (\langle \hat \sigma_z^i \rangle + 1).
\end{equation}
Defining the observable \(\hat Y_i = \hat \sigma_y^i \cos \varphi_i   + \hat \sigma_x^i  \sin \varphi_i\), we can solve this EOM in the steady state (\(\partial_t \langle \hat{\sigma}^j_z \rangle=0\)), yielding:
\begin{equation}
\label{eq:ss-yz}
     \langle \hat Y_i \rangle = \frac{\kappa}{\omega_i} (\langle \hat \sigma_z^i \rangle + 1).
\end{equation}
We see that in the steady state, this exact relation tells us that $\hat{Y}_i$ expectation values are completely determined by $\hat \sigma_z^i$ expectation values.  
Note that the observable \(\hat Y_i\) is orthogonal in the transverse plane to the driving term in the Hamiltonian Eq.~\eqref{eq:dtfim-heff}, given by \(\omega_i \hat X_i/2\) with \(\hat X_i \equiv \hat \sigma_x^i \cos \varphi_i - \hat \sigma_y^i \sin \varphi_i\).
For the case of real drives in the main text (i.e.~\(\varphi_j = 0\)), these observables simply become \(\hat Y_i \to \hat \sigma_y^i\) and \(\hat X_i \to \hat \sigma_x^i\).

Combining Eqs.~\eqref{eq:logN} and~\eqref{eq:ss-yz} lets us define a scalar function \(\mathcal G(\vec \lambda) \) such that we can express \(\langle \hat Y_i \rangle\) as its parametric derivative with respect to~\({\rm Re}\lambda_i\).  We have:
\begin{align}
\mathcal G \equiv \frac{\kappa}{2}\left( \log \mathcal N + 2\sum_j (\rm Re \lambda_j)\right) \quad \implies  \frac{\partial \mathcal G}{\partial ({\rm Re} \lambda_i)}    
= \kappa \left( \langle \hat \sigma_z^i \rangle + 1 \right) 
= \omega_i \langle \hat Y_i \rangle.
\end{align}
Finally,  since we have \(\partial_{{\rm Re} \lambda_i} = \omega_i\partial_{\omega_i}\), this reduces to
\begin{equation}
\label{eq:tfim-yi-scalar-derivative}
 \langle \hat Y_i \rangle = \frac{\partial \mathcal G}{\partial \omega_i} \qquad \mathcal G = \frac{\kappa}{2} \log \left(\mathcal N \prod_j \omega_j^2 \right).
\end{equation}
Then, since second derivatives of this scalar \(\mathcal G\) trivially commute, we have
\begin{equation}
\label{eq:gen-tfim-reciprocity}
    \frac{\partial^2 \mathcal G}{\partial \omega_j \partial \omega_i} =  \frac{\partial^2 \mathcal G}{\partial \omega_i \partial \omega_j} \implies \frac{\partial \langle \hat Y_i \rangle}{\partial \omega_j} = \frac{\partial \langle \hat Y_j \rangle}{\partial \omega_i}
\end{equation}
and we are done.
We now note a minor caveat: in the parametrization of \(\Omega_j\) in Eq.~\eqref{eq:lambda-tfim}, we assume \(\omega_j \neq 0\).
As such, Eq.~\eqref{eq:gen-tfim-reciprocity} only presently holds for this case.
We can however show by a simple continuity argument that it holds for all \(\omega_j \in \mathbb R\).
First, we note the observable \(\langle\hat Y_i\rangle\) is a regular function of the drives \(\omega_i\) (since \(\langle\hat \sigma_z^i\rangle + 1 \sim \mathcal O(\omega_i^2)\)).
Then, since the Liouvillian gap is non-zero and steady state is unique for any finite \(N\)~\cite{roberts2023exact}, \(\hat \rho_{\rm ss}\) is an analytic function of the drives.
Therefore, the parametric derivatives of \(\langle \hat Y_i \rangle\) in  Eq.~\eqref{eq:gen-tfim-reciprocity} are continuous, and since the equality holds for the open, dense set \(\omega_j \neq 0, \omega_j \in \mathbb R\), by continuity it must hold for all \(\omega_j \in \mathbb{R}\).

Eq.~\eqref{eq:gen-tfim-reciprocity} is the most general form of the susceptibility symmetry for arbitrary complex drives \(\vec{\Omega}\).
In the main text, we work in the case where the drives are real (i.e.~\(\Omega_j = \omega_j\)), and this reduces to exactly Eq.~\eqref{eq:tfim-reciprocity}.

\section{Proof of susceptibility symmetry in bosonic lattice with quadratic driving}
\label{section:bosonic-proof}
In this section, we prove the susceptibility symmetry Eq.~\eqref{eq:bosonic-recRelation}, enabling EP in the bosonic lattice model (see Sec.~\ref{section:bosonic-lattice} for details about the model).
The proof is structured similarly to the one in App.~\ref{section:dtfim-reciprocity-derivation} for the DTFIM.
We will first state Theorem~\ref{theorem:bosonic}, which relates steady states of different instances of the bosonic model, corresponding to different choices of the two-photon drive amplitudes (encoded in the symmetric matrix $M_{ij}$).
Similar to Theorem~\ref{theorem:tfim}, the central idea here is that hTRS lets us extend a trivial gauge covariance of the Lindbladian to a much more general statement.
Then, in App.~\ref{app:bosonic-proof-final}, we use Theorem~\ref{theorem:bosonic} to prove the susceptibility symmetry required to do EP in this system.

The most general description of the system is given by the Lindbladian Eq.~\eqref{eq:bosonic-lind}. 
The effective Hamiltonian is given by
\begin{equation}
    \label{eq:bosonic-heff}
    \hat H_{\rm eff} (M) =   \frac{U}{N} \left( \sum_j \hat{n}_j\right)^2 - \left(\Delta + \frac{i\kappa}{2}\right) \sum_j \hat{n}_j + \sum_{i,j}  \left(M_{ij}\hat{a}_i^{\dag}\hat{a}_j^{\dag} + \text{h.c.}\right)
\end{equation}
where \(M\) is an arbitrary complex symmetric matrix which we refer to as the driving matrix.
As in App.~\ref{section:dtfim-reciprocity-derivation}, we will work with the formulation of hTRS described in App.~\ref{appendix:htrs}, based on the steady state purification \(\ket{\Phi}\) and doubled system operators.
Since we are interested in how the steady state transforms as we change the driving matrix \(M\), we explicitly index the purification corresponding to a particular driving matrix as \(\ket{\Phi} \equiv \ket{\Phi (M)}\) for clarity.

Similar to the DTFIM proof, we introduce a set of operators on the doubled-system Hilbert space that will generate transformations of $M$.  In this case, there is a generator for every pair of lattice sites $p,q$ given by the hopping operator \(\hat Z_{pq} \equiv \frac{1}{2}\left( \hat a^{\dag}_{p,A} \hat a_{q, A} + \hat a^{\dag}_{p,B} \hat a_{q,B}\right)\).  Note that since \(\hat Z_{pq}^{\dag} = \hat Z_{qp}\), only the diagonal generators \(\hat Z_{pp}\) are Hermitian; these are analogous to the charge operators \(\hat Z_j\) in the DTFIM (see App.~\ref{section:dtfim-reciprocity-derivation}).

It will be useful to understand how the exponential of this operator transforms the bosonic lowering operator \(\hat a_{i,\alpha}\) (where \(\alpha \in [A,B]\) indexes the subsystem we use).
\begin{lemma}
\label{lemma:bosonic-rotation}
    For an arbitrary matrix \(\lambda \in \mathbb{C}^{N \times N}\), define \(S \equiv e^{\lambda/2}\). Then, we have 
    \begin{equation}
    e^{-\sum_{pq}\lambda_{pq} \hat Z_{pq}} \hat a_{i,\alpha} e^{\sum_{pq}\lambda_{pq} \hat Z_{pq}} = \sum_{j}S_{ij} \hat a_{j, \alpha}.
    \end{equation}
\end{lemma}
\begin{proof}
From Hausdorff's lemma, we have:
\begin{equation}
e^{-\sum_{pq}\lambda_{pq} \hat Z_{pq}} \hat a_{i,\alpha} e^{\sum_{pq}\lambda_{pq} \hat Z_{pq}} = \hat a_{i,\alpha} -\Big[ \sum_{pq}\lambda_{pq} \hat Z_{pq},\hat a_{i,\alpha} \Big] + \frac{1}{2}\Big[\sum_{pq}\lambda_{pq} \hat Z_{pq}, \Big[\sum_{pq}\lambda_{pq} \hat Z_{pq},\hat a_{i,\alpha} \Big]\Big] + \dots
\end{equation}
Using the standard bosonic commutation relations, we have \(\Big[\hat a_{i,\alpha}, \sum_{pq}\lambda_{pq} \hat Z_{pq} \Big] = \frac{1}{2}\sum_{j}\lambda_{ij} \hat a_{j,\alpha}\).
Then, the series reduces to 
\begin{equation}
    e^{-\sum_{pq}\lambda_{pq} \hat Z_{pq}} \hat a_{i,\alpha} e^{\sum_{pq}\lambda_{pq} \hat Z_{pq}} = \sum_{j}\left(\sum_{n \geq 0} \frac{(\lambda/2)^n}{n!}\right)_{ij} \hat{a}_{j,\alpha} = \sum_{j}(e^{\lambda/2})_{ij} \hat a_{j,\alpha}.
\end{equation}

\end{proof}

Note that we can similarly solve for how the raising operator changes under this transformation, using the commutation relation \(\Big[\hat a^{\dag}_{i,\alpha}, \sum_{pq}\lambda_{pq} \hat Z_{pq} \Big] =-\frac{1}{2}\sum_{j}\lambda_{ji} \hat a^{\dag}_{j,\alpha}\). Explicitly, we have:
\begin{equation}
 e^{-\sum_{pq}\lambda_{pq} \hat Z_{pq}} \hat a^{\dag}_{i,\alpha} e^{\sum_{pq}\lambda_{pq} \hat Z_{pq}} = \sum_j \left(e^{-\lambda^T/2} \right)_{ij} \hat a^{\dag}_{j, \alpha} = \sum_{j}(S^{-T})_{ij} \hat a^\dag_{j, \alpha}
    \end{equation}
where \(S^{-T}\) is shorthand for \((S^{-1})^T\).
Note that when \(\lambda\) is chosen to be anti-Hermitian, \(S\) is a unitary matrix, and Lemma~\ref{lemma:bosonic-rotation} describes a simple \(U(N)\) gauge transformation. However in the general case, where \(S\) is any invertible complex matrix, this transformation also rescales the magnitudes of the driving matrix elements.

We now state our central theorem.

\begin{theorem}
    \label{theorem:bosonic}
    Consider a \(N\) mode bosonic system described by the Lindbladian Eq.~\eqref{eq:bosonic-lind}, with driving characterized by a complex symmetric matrix \(M\). Let \(\ket{\Phi(M)}\) be a purification (not necessarily normalized) of its steady state that satisfies the hTRS conditions Eq.~\eqref{eq:darkstate-1} and Eq.~\eqref{eq:darkstate-2}.
    Then, a second bosonic system with driving matrix \(M' = S M S^T\) for any \(S = e^{\lambda/2}\) with \(\lambda \in \mathbb{C}^{N \times N}\) (and all other system parameters identical to the first) has a purification (also not necessarily normalized) of its steady state \(\ket{\Phi(M')}\) given by \(\ket{\Phi(M')} = e^{\sum_{pq} \lambda_{pq}\hat Z_{pq}}\ket{\Phi(M)}\) that satisfies the hTRS conditions. 
\end{theorem}

This theorem establishes the relationship between the steady states of different instances of the bosonic model, as Theorem~\ref{theorem:tfim} does for the DTFIM.
Furthermore, there is also a limit in which Theorem~\ref{theorem:bosonic} becomes trivial.
Explicitly, when \(\lambda\) is anti-Hermitian, and consequently \(S\) is a unitary matrix, Theorem~\ref{theorem:bosonic} is a statement of the covariance of the Lindbladian and its steady state under a \(U(N)\) gauge transformation.
Any unitary matrix \( U\) transforms the bosonic lattice Lindbladian with driving matrix \(M\) to a new bosonic lattice Lindbladian with driving matrix \(M' = UMU^T\), and the steady state is transformed by the corresponding unitary operator. 
This covariance does not require a notion of hTRS.
However, in full generality, Theorem~\ref{theorem:bosonic} gives us the relation between steady states for two driving matrices related by {\it any} complex matrix \(\lambda\), with \(S = e^{\lambda/2}\): this is a much more powerful relation that is enabled by hTRS.

\subsection{Proof of Theorem~\ref{theorem:bosonic} using hTRS}
\label{app:bosonic-htrs}

Let \(\ket{\Phi(M)}\) be a purification of the bosonic lattice model steady state with driving matrix \(M\) that necessarily satisfies the hTRS conditions Eq.~\eqref{eq:darkstate-1} and Eq.~\eqref{eq:darkstate-2}.
Let us also define \(\ket \phi \equiv e^{\sum_{pq} \lambda_{pq} \hat Z_{pq}} \ket{\Phi(M)}\).
We now show that for arbitrary \(\lambda \in \mathbb{C}^{N \times N}\), \(\ket \phi\) satisfies these hTRS conditions for a driving matrix \(M' = SMS^T\). This \(\ket \phi\) is then necessarily a purification of the bosonic lattice steady state with this driving matrix, i.e.~\(\ket \phi = \ket{\Phi(M')}\).

By Lemma~\ref{lemma:bosonic-rotation}, we have
\begin{equation}
    e^{-\sum_{pq}\lambda_{pq}\hat Z_{pq}}(\hat a_{i,A} - \hat a_{i, B}) \underbrace{e^{\sum_{pq}\lambda_{pq}\hat Z_{pq}} \ket{\Phi(M)}}_{= \ket{\phi}}  = \sum_j S_{ij} \underbrace{\left( \hat a_{j,A} - \hat a_{j, B}\right) \ket{\Phi(M)}}_{=0}.
\end{equation}
The RHS is trivial by Eq.~\eqref{eq:darkstate-1} for \(\ket{\Phi(M)}\).
Then, we can multiply both sides on the left by \(e^{\sum_{pq}\lambda_{pq}\hat Z_{pq}}\), yielding
\begin{equation} 
\label{eq:darkstate-1-phi}
    (\hat a_{i,A} - \hat a_{i, B}) \ket{\phi} = 0.
\end{equation}
We have thus proved the first hTRS relation for the state \(\ket \phi\).
For the second condition, it will be useful to establish a relation analogous to Lemma~\ref{lemma:bosonic-rotation} for the single system effective Hamiltonian with the driving matrix \(M'\).
We do this term by term. 
First, we show that the total occupation is invariant under such a transformation:
\begin{equation}
    e^{-\frac{1}{2}\sum_{pq}\lambda_{pq}\hat a^\dag_p \hat a_q} \left(\sum_{j} \hat n_j\right) e^{\frac{1}{2}\sum_{pq}\lambda_{pq}\hat a^\dag_p \hat a_q} = \sum_{jpq}(S^{-T})_{jp}S_{jq}\,\hat a^\dag_p\hat a_q
=\sum_{pq}(S^{-1}S)_{pq}\,\hat a^\dag_p\hat a_q= \sum_{j} \hat n_j.
\end{equation}
The first driving term, with pair raising operators, simply shifts \(M'\) to \(M\) under this transformation:
\begin{align}
    e^{-\frac{1}{2}\sum_{pq}\lambda_{pq}\hat a^\dag_p \hat a_q} \left(\sum_{ij} M'_{ij} \hat a^\dag_i \hat a^{\dag}_j\right) e^{\frac{1}{2}\sum_{pq}\lambda_{pq}\hat a^\dag_p \hat a_q} & = \sum_{ij}M'_{ij}\left(\sum_k(S^{-T})_{ik}\hat a^\dag_k\right)
\left(\sum_l(S^{-T})_{jl}\hat a^\dag_l\right) \\
    & = \sum_{ijkl}(S^{-1})_{ki}M'_{ij} (S^{-T})_{jl}\hat a^\dag_k \hat a^\dag_l\\ 
    & = \sum_{kl} (S^{-1} (SMS^T)S^{-T})_{kl} \hat a^\dag_k \hat a^\dag_l \\
    & = \sum_{kl} M_{kl} \hat a^\dag_k \hat a^\dag_l.
\end{align}
However, as in the DTFIM, the lowering terms change non-trivially:
\begin{align}
    e^{-\frac{1}{2}\sum_{pq}\lambda_{pq}\hat a^\dag_p \hat a_q} \left(\sum_{ij} M'^{*}_{ij} \hat a_i \hat a_j\right) e^{\frac{1}{2}\sum_{pq}\lambda_{pq}\hat a^\dag_p \hat a_q} & = \sum_{ijkl} S_{ik}M'^*_{ij} S_{jl} \hat a_k \hat a_l \\
    & = \sum_{kl} (S^T M'^* S)_{kl} \hat a_k \hat a_l \\
    & = \sum_{kl} (S^T S^*M^*S^{*T} S)_{kl} \hat a_k \hat a_l.
\end{align}
The matrix \(S^T S^*M^*S^{*T} S\) does not generically equal \(M^*\), and so this provides an extra `error' term when transforming the effective Hamiltonian.
We can explicitly define this `error' matrix \(R \equiv S^T S^*M^*S^{*T} S - M^*\), and note crucially that it is a symmetric matrix, since \(M^*\) is symmetric. We can now use the above work to show:
\begin{align}
\label{eq:bosonic-heff-transform}
    e^{-\sum_{pq} \lambda_{pq} \hat Z_{pq}} \left(\hat H_{\rm eff,A} (M') - \hat H_{\rm eff,B} (M') \right) e^{\sum_{pq} \lambda_{pq} \hat Z_{pq}} & = \hat H_{\rm eff,A} (M) - \hat H_{\rm eff,B} (M) \nonumber \\ & + \sum_{pq} R_{pq} \left(\hat a_{p, A} \hat a_{q,A} - \hat a_{p, B} \hat a_{q,B} \right).
\end{align}
Note that this `error' matrix \(R\) is trivial for all \(M\) exactly in the case where \(S\) is unitary, since for unitary \(S\), we would have \(S^T = (S^{-1})^*\), and \(S^T S^*M^*S^{*T} S = M^*\).
While for arbitrary \(S\) this error term is generically non-zero, it vanishes when acting on the state \(\ket{\Phi(M)}\):
\begin{align}
    \sum_{pq} R_{pq} \left(\hat a_{p, A} \hat a_{q,A} - \hat a_{p, B} \hat a_{q,B} \right) \ket{\Phi(M)} = \sum_{pq} R_{pq} \left(\hat a_{p, A} + \hat a_{p,B}\right) \underbrace{\left( \hat a_{q,A} - \hat a_{q, B}  \right) \ket{\Phi(M)}}_{=0}
\end{align}
where the cross terms vanish in the sum on the RHS since \(R\) is symmetric.
Crucially, we have that the RHS vanishes from the hTRS condition Eq.~\eqref{eq:darkstate-1}.
Indeed, combined with the other hTRS condition Eq.~\eqref{eq:darkstate-2}, this implies that the entire RHS of Eq.~\eqref{eq:bosonic-heff-transform} annihilates \(\ket{\Phi(M)}\).
Then, applying both sides of Eq.~\eqref{eq:bosonic-heff-transform} to the state \(\ket{\Phi(M)}\) yields:
\begin{equation}
    e^{-\sum_{pq} \lambda_{pq} \hat Z_{pq}} \left(\hat H_{\rm eff,A} (M') - \hat H_{\rm eff,B} (M') \right) \underbrace{e^{\sum_{pq} \lambda_{pq} \hat Z_{pq}} \ket{\Phi(M)}}_{= \ket \phi}  = 0.
\end{equation}
Multiplying both sides by \( e^{\sum_{pq} \lambda_{pq} \hat Z_{pq}}\), we have
\begin{equation}
\label{eq:darkstate-2-phi}
\left(\hat H_{\rm eff,A} (M') - \hat H_{\rm eff,B} (M') \right)\ket \phi = 0.
\end{equation}
Thus, in Eqs.~\eqref{eq:darkstate-1-phi} and~\eqref{eq:darkstate-2-phi} we have shown that the state \(\ket \phi\) satisfies the hTRS conditions Eqs.~\eqref{eq:darkstate-1} and~\eqref{eq:darkstate-2} for the driving matrix \(M'\). The state \(\ket \phi\) is then necessarily the purification of the bosonic lattice steady state for this driving matrix, i.e.~\(\ket \phi = \ket{\Phi(M')}\), and we have proven Theorem~\ref{theorem:bosonic}.

\subsection{Susceptibility symmetry in bosonic lattice}
\label{app:bosonic-proof-final}

We now use Theorem~\ref{theorem:bosonic} to prove the susceptibility symmetry in the bosonic model.
We will proceed as follows.
In App.~\ref{app:bosonic-diagonal}, we will use Theorem~\ref{theorem:bosonic} to prove a version of the susceptibility symmetry for the simple case of diagonal driving matrices.
 Then, in App.~\ref{app:bosonic-full}, we will show how this result generalizes to arbitrary symmetric driving matrices.

\subsubsection{Diagonal driving matrices}
\label{app:bosonic-diagonal}

Let us fix all parameters of the bosonic lattice Lindbladian except the driving matrix.
We assume the driving matrix \(D\) to be a diagonal matrix: \(D \equiv \text{diag}(\vec d)\) with non-zero entries \(d_p \neq 0\).
We denote the  corresponding family of purifications of the bosonic lattice steady states as \(\ket{\Phi(D)}\).
As in the case of the DTFIM, we can define a normalization for this family of states \(\mathcal N(D) = \braket{\Phi(D)}{\Phi(D)}\), and our goal is to parameterize this in terms of some reference state.
We will find that this simple case of diagonal driving matrices is exactly analogous to the DTFIM, with diagonal entries \(d_p\) playing the role of the drives \(\Omega_j\).
Indeed, the work in this section closely mirrors App.~\ref{app:dtfim-proof}.

First, we note that the effective Hamiltonian for diagonal driving matrices is much simplified:
\begin{equation}
\label{eq:heff-diagonal}
    \hat H_{\rm eff} (D) =  \frac{U}{N} \left( \sum_j \hat{n}_j\right)^2 - \left(\Delta + \frac{i\kappa}{2}\right) \sum_j \hat{n}_j + \sum_{p} d_p \left(\hat{a}_p^{\dag 2}+ \hat a_p^2\right).
\end{equation}

Since we are working with diagonal matrices, a suitable reference state is \(\ket{\Phi(D_0)}\) where \(D_0\) is a uniform diagonal matrix \(D_0 = \text{diag}(\bar{d}, \bar{d},\dots)\) with \(\bar d > 0\).
Then, there trivially always exists some diagonal \(S\) such that \(D = SD_0S^T\). Recall that \(S \equiv e^{\lambda/2}\), and we can therefore parameterize \(\vec{d}\) in terms of \(\lambda\): 
\begin{equation}
\label{eq:lambda-bosonic}
    \lambda_{pq} = \delta_{pq}\log(d_p/\bar{d}).
\end{equation}
We emphasize that \(\lambda\) is now a diagonal matrix.
Then the family of steady state purifications \(\ket{\Phi(D)}\) can be written as:
\begin{equation}
    \ket{\Phi(D)} = e^{\sum_p \lambda_{pp} \hat{Z}_{pp}}\ket{\Phi(D_0)}.
\end{equation}
We can now express our normalization in terms of the reference state as
\begin{equation}
    \mathcal N (D) = \bra{\Phi(D_0)}e^{2\sum_{p} {\rm Re}\lambda_{pp}\hat Z_{pp}}\ket{\Phi(D_0)}.
\end{equation}
This is where the advantage of working with the diagonal driving matrix becomes clear: since the \(\lambda\) matrix is diagonal, all terms in the exponential commute with each other. 
The generators \(\hat Z_{pp}\) are Hermitian charges, and this expression is now exactly analogous to Eq.~\eqref{eq:norm-tfim} in the DTFIM case, with the \(d_p\) playing the role of the drives \(\Omega_j\).
We can now differentiate with respect to \({\rm Re}\lambda_{pp}\), yielding
\begin{equation}
    \frac{\partial \mathcal N(D)}{\partial {\rm Re}\lambda_{pp}} = 2\bra{\Phi(D)} \hat Z_{pp} \ket{\Phi(D)}.
\end{equation}
Note now that from Lemma~\ref{lemma:z}, we know that \(\bra{\Phi(D)} \hat Z_{pp} \ket{\Phi(D)} =\bra{\Phi(D)} \hat n_{p,A} \ket{\Phi(D)}\). Then we can express normalized expectation values on the single system by dividing both sides by \(\mathcal N\), yielding
\begin{equation}
\label{eq:bosonic-logn}
    \frac{\partial \log \mathcal N(D)}{\partial {\rm Re}\lambda_{pp}} = 2\langle \hat n_p \rangle_D
\end{equation}
where we use the subscript \(\langle.\rangle_D\) to denote the expectation value in the state \(\ket{\Phi(D)}\). 
We have thus shown that the normalization \(\mathcal N(D)\) is the generating function for \(\hat n_p\) expectation values in the steady state \(\ket{\Phi(D)}\).

As in the case of the DTFIM, we now use the Heisenberg equation of motion for \(\hat n_p\) to connect this to our observables of interest:
\begin{equation}
    \partial_t \langle \hat n_p \rangle_D = 2id_p\left( \langle \hat a_p^2\rangle_D - \langle\hat{a}^{\dag 2}_p  \rangle_D\right) - \kappa \langle \hat n_p \rangle_D.
\end{equation}
Defining the observable \(\hat Y'_p \equiv \frac{1}{2i}\left(\hat a_p^2 -\hat{a}^{\dag 2}_p\right)\), we can solve this equation in the steady state (\(\partial_t  \langle \hat n_p \rangle_D = 0\)) yielding
\begin{equation}
\label{eq:bosonic-yd}
    \langle \hat Y'_p \rangle_D = -\frac{\kappa}{4 d_p}\langle \hat n_p \rangle_D.
\end{equation}
Combining Eq.~\eqref{eq:bosonic-logn} and Eq.~\eqref{eq:bosonic-yd}, we can define a scalar \(\mathcal G \) such that we can express \(\langle \hat Y'_p\rangle_D\) as its parametric derivative with respect to \({\rm Re}\lambda_{pp}\):
\begin{equation}
\label{eq:bosonic-g}
    \mathcal G \equiv -\frac{\kappa}{8} \log \mathcal N (D) \quad \implies \frac{\partial\mathcal G}{\partial {\rm Re}\lambda_{pp}} = - \frac{\kappa}{4}\langle \hat n_p \rangle_D = d_p \langle \hat Y'_p \rangle_D.
\end{equation}
By our parametrization Eq.~\eqref{eq:lambda-bosonic} we have \(\partial_{ {\rm Re}\lambda_{pp}} = d_p \partial_{d_p}\). Then, from Eq.~\eqref{eq:bosonic-g} we have:
\begin{equation}
\label{eq:yp-D}
    \langle \hat Y'_p\rangle_D = \frac{\partial \mathcal G}{\partial d_p}.
\end{equation}
This is the crucial result of this section.
This equation is exactly analogous to Eq.~\eqref{eq:tfim-yi-scalar-derivative}, where we have now expressed an observable of interest as a parametric derivative with respect to a scalar \( \mathcal G \sim \log \mathcal N\).
In the DTFIM case, this is the end of the proof: from the commutativity of the second derivatives one trivially then derives the susceptibility symmetry.

In the bosonic case however, we have some work left to do; we have only established this relation for diagonal driving matrices.
Next, we will generalize this result by relating a family of states \(\ket{\Phi(M)}\) for arbitrary driving matrices \(M\) to a corresponding family of states with diagonal driving matrices \(\ket{\Phi(D_M)}\).

\subsubsection{Susceptibility symmetry}
\label{app:bosonic-full}
In this section, we will exploit the results we have established for diagonal driving matrices to prove a susceptibility symmetry for a general family of steady states \(\ket{\Phi(M)}\) with \(M\) denoting arbitrary complex symmetric driving matrices.
The crucial step here is to employ the \(U(N)\) gauge covariance of the Lindbladian, i.e.~use Theorem~\ref{theorem:bosonic} for a unitary \(S\) to express the normalization \(\mathcal N (M) \equiv \braket{\Phi(M)}{\Phi(M)}\) in an equivalent and more convenient form.
Note that throughout this section we treat \(M_{ij}\) and \(M^*_{ij}\) as independent variables.

Explicitly, for a complex symmetric matrix \(M\), we can always write the decomposition~\cite{roberts2023competition}:
\begin{equation}
\label{eq:M-takagi}
    M = VD_MV^T
\end{equation}
where \(V\) is a unitary matrix and \(D_M\) is a real, diagonal matrix.
This decomposition is not unique: we have a freedom in choosing the signs of the diagonal entries of \(D_M\), which we denote \(d_p^M\).
One such choice is the Takagi decomposition, which gives unitary \(V\) and real diagonal \(D_M\) with {\it nonnegative} entries.
Critically, we do {\it not} use a Takagi decomposition, and do not restrict the \(d_p^M\) to be positive.
The Takagi singular values are then precisely \(|d_p^M|\), and the corresponding signs get absorbed into the definition of \(V\).

At this stage, we restrict \(M\) to be an invertible matrix, so that the Takagi singular values \(|d_p^M| > 0\), and we can invoke the results of App.~\ref{app:bosonic-diagonal} for the diagonal matrix \(D_M\). 
Later, we will be able to loosen this assumption with a simple continuity argument.

We can now relate steady states of driving matrices \(M\) and \(D_M\) using the trivial form of Theorem~\ref{theorem:bosonic}: the relation is simply a covariance under a \(U(N)\) gauge transformation, i.e.~\(\ket{\Phi(M)} = \hat U \ket{\Phi(D_M)}\) for some unitary \(\hat U \).
Since \(V \equiv e^{\lambda/2}\), we have \(\lambda = 2 \log V\), and this unitary is explicitly
\begin{equation}
\label{eq:bosonic-u}
    \hat U = \exp\left(2\sum_{pq}(\log V)_{pq} \hat Z_{pq} \right).
\end{equation}
Therefore, we have 
\begin{equation}
\label{eq:bosonic-nd-nm}
    \mathcal N (M) = \braket{\Phi(M)}{\Phi(M)} = \braket{\Phi(D_M)}{\Phi(D_M)} = \mathcal{N}(D_M).
\end{equation}
This is the key insight we will use: the normalization is independent of the basis defined by \(V\), and depends only on the \(d^M_p\), and specifically only the magnitudes \(|d_p^M|\).
This will allow us to relate parametric derivatives with respect to the \(d_p^M\) (like in Eq.~\eqref{eq:yp-D}) in terms of parametric derivatives with respect to matrix elements \(M_{ij}\).

Explicitly, our goal is to express \(\partial \mathcal G/ \partial M^*_{ij}\) (\(\mathcal G\) defined in Eq.~\eqref{eq:bosonic-g}) in terms of \(\partial \mathcal G/ \partial d_p^M\), so that we can invoke the crucial result Eq.~\eqref{eq:yp-D}. 
We can use a chain rule to convert between parametric derivatives of \(d_p^M\) and \(M_{ij}\).
Note that this requires the \(d^M_p\) to be a smooth function of the matrix elements \(M_{ij}\), which is true only when the \(|d_p^M|\) are non-degenerate.
Therefore, we add a second restriction on \(M\) (which again we will relax later via a continuity argument): \(M\) must be an invertible matrix with non-degenerate Takagi singular values.
Then, we have:
\begin{align}
    \frac{\partial \mathcal G}{\partial M^*_{ij}} & = \sum_p \frac{\partial d^M_p}{\partial M^*_{ij}} \frac{\partial \mathcal G}{\partial d^M_p}  \\
    & = \frac{(2 - \delta_{ij})}{2}\sum_p V_{ip}V_{jp} \langle \hat Y'_p \rangle_{D_M}.
    \end{align}
In writing the chain rule, we have exploited the fact that \(\mathcal G\) is completely determined by \(\mathcal N(D_M)\), thus depending only on the \(d_p^M\), and not on the matrix elements of \(V\). 
The factor of \((2 - \delta_{ij})\) (where \(\delta_{ij}\) is the Kronecker delta) is to account for double counting of off-diagonal pairs as discussed in the main text.

The next step is to express the observable \(\langle \hat Y'_p\rangle_{D_M}\) in terms of expectation values in our states of interest \(\ket{\Phi(M)}\).
We will use the unitary \(\hat U\) (see Eq.~\eqref{eq:bosonic-u}) to convert between steady states with driving matrices \(M\) and \(D_M\).
Explicitly, for any operator \(\hat O\), we have \(\langle \hat O \rangle_{D_M} = \langle \hat U \hat O \hat U^{\dag} \rangle_M\).
Using Lemma~\ref{lemma:bosonic-rotation}, we have \(\hat U \hat a_{i, \alpha} \hat U^{\dag} = \sum_{j} V^*_{ji} \hat a_{j,\alpha}\) and \(\hat U \hat a^{\dag}_{i, \alpha} \hat U^{\dag} = \sum_{j} V_{ji} \hat a^\dag_{j,\alpha}\).
Then, we express the observable \(\langle \hat Y'_p\rangle_{D_M}\) as
\begin{equation}
    \langle \hat Y'_p\rangle_{D_M} = \frac{1}{2i}\left(\sum_{kl} \left( V^*_{kp} V^*_{lp} \langle \hat a_k \hat a_l\rangle_M - V_{kp} V_{lp} \langle \hat a^\dag_k \hat a^\dag_l\rangle_M\right)\right).
\end{equation}
Now, we exploit the fact that for a diagonal driving matrix \(D_M\), the steady state has no pair coherences between modes. The Lindbladian (see effective Hamiltonian Eq.~\eqref{eq:heff-diagonal}) is invariant under the transformation \(\hat a_p \to - \hat a_p\) for each mode \(p\), i.e.~it has a \(\mathbb{Z}^N_2\) symmetry. Therefore, we have \(\langle \hat a_p \hat a_q \rangle_D = \delta_{pq} \langle \hat a^2_p\rangle_D\). Consequently, we have \(\langle \hat Y_{pq} \rangle_{D_M} = \delta_{pq}\langle \hat Y'_{p}\rangle_{D_M}\) (with \(\hat Y_{pq}\) as defined in Eq.~\eqref{eq:yij}). Then, we have:
\begin{align}
\label{eq:dg-1}
\frac{\partial \mathcal G}{\partial M^*_{ij}} & = \frac{(2 - \delta_{ij})}{2}\sum_{pq} V_{ip}V_{jq} \langle \hat Y_{pq} \rangle_{D_M} \\ 
 &= \frac{(2 - \delta_{ij})}{4i}\left(
       \sum_{kl}\left(\sum_p V_{ip}V^*_{kp}\right)\left(\sum_q V_{jq}V^*_{lq}\right)
       \langle \hat a_k \hat a_l \rangle_M
     - \sum_{kl}\left(\sum_p V_{ip}V_{kp}\right)\left(\sum_q V_{jq}V_{lq}\right)
       \langle \hat a^{\dag}_k \hat a^{\dag}_l \rangle_M \right) \\
& = \frac{(2 - \delta_{ij})}{4i}\left(\langle \hat a_i \hat a_j \rangle_M - \sum_{kl} W_{ik} W_{jl} \langle \hat a^{\dag}_k \hat a^{\dag}_l\rangle_M\right) \equiv \langle \hat{\mathcal{Y}}_{ij} \rangle
\end{align}
where we critically define \(W = VV^T\), i.e.~\(W_{ik} = \sum_p V_{ip}V_{kp}\). 

Then, since second derivatives of the scalar \(\mathcal G\) trivially commute, we have
\begin{equation}
\label{eq:yij-ykl-general}
    \frac{\partial^2 \mathcal G}{\partial M^*_{kl} \partial M^*_{ij}} = \frac{\partial^2 \mathcal G}{\partial M^*_{ij} \partial M^*_{kl}}  \quad \implies \frac{\partial\langle \hat{\mathcal{Y}}_{ij}\rangle}{\partial M^*_{kl}} = \frac{\partial \langle  \hat{\mathcal{Y}}_{kl}\rangle}{\partial M^*_{ij}}
\end{equation}
for a complex symmetric invertible \(M\) with non-degenerate Takagi singular values.
This is the most general statement we can make for a complex, symmetric \(M\).

At first glance, Eq.~\eqref{eq:yij-ykl-general} seems like the kind of susceptibility symmetry we are looking for.
However, the operator \(\hat{\mathcal{Y}}_{kl}\) is not a Hermitian observable, and implicitly depends on the driving matrix \(M\) through \(V\), which makes it ill-suited for the purposes of doing EP.
Since our training parameters are the matrix elements \(M_{ij}\), this operator \(\hat{\mathcal{Y}}_{kl}\) would change continuously during training, making it infeasible to extract gradients accurately.
Furthermore, it would require {\it global} information through the matrix \(M\), which defeats the purpose of using a {\it local} learning rule like EP.

Therefore, we must find a way to restrict the statement such that  we recover a fixed Hermitian observable that does not depend on \(V\).
By inspection, this occurs when \(V\) is a real, orthogonal matrix, so \(W = VV^T = \mathbb{I}\).
This is precisely the case where \(M\) is a real symmetric matrix: for such an \(M\) the decomposition in Eq.~\eqref{eq:M-takagi} is precisely an eigendecomposition where \(V\) is a real, orthogonal matrix, and the entries of the diagonal matrix \(D_M\) are the corresponding eigenvalues. 
As a result, we will restrict our focus to this case of real, symmetric driving matrix \(M\). 
In this case, for non-degenerate eigenvalues \(d_p^M\), we have \(\partial d_p^M/\partial M_{ij} = (2 - \delta_{ij})V_{ip} V_{jp}\).
Then, similar to Eq.~\eqref{eq:dg-1} we have:
\begin{equation}
    \frac{\partial \mathcal G}{\partial M_{ij}} = (2 - \delta_{ij})\sum_{pq} V_{ip} V_{jq} \langle \hat Y_{pq}\rangle_{D_M} = (2 - \delta_{ij})\langle \hat Y_{ij} \rangle
\end{equation}
where in the second step we use \(\langle \hat a_p \hat a_q\rangle_{D_M} = \sum_{kl} V_{kp}V_{lq} \langle \hat a_k \hat a_l \rangle_M\) and the fact that \(V V^T = \mathbb{I}\). 
Then, from commutativity of second derivatives we have:
\begin{equation}
\label{eq:yij-sym}
    (2 - \delta_{ij})\frac{\partial \langle \hat{Y}_{ij} \rangle}{\partial M_{kl}} =  (2 - \delta_{kl})\frac{\partial \langle \hat{Y}_{kl} \rangle}{\partial M_{ij}}
\end{equation}
for a real, symmetric, {\it invertible} \(M\) with {\it non-degenerate eigenvalues}.
Crucially, \(\hat Y_{ij}\) is a {\it fixed} observable independent of \(M\), and is well suited to perform EP.

We now make a continuity argument to extend Eq.~\eqref{eq:yij-sym} to arbitrary real, symmetric \(M\).
In particular, \(\langle \hat Y_{ij} \rangle\) depends on the driving matrix \(M\) only through the steady state \(\hat \rho_{\rm ss}\). The steady state varies smoothly in \(M_{ij}\), since it is unique and the Liouvillian gap is non-zero for finite \(N\). 
Therefore, the parametric derivatives in Eq.~\eqref{eq:yij-sym} are continuous, and since the equality holds on the open, dense set of real symmetric \(M\) with non-degenerate, non-zero eigenvalues, by continuity it holds for all real symmetric \(M\).
We have thus proved Eq.~\eqref{eq:bosonic-recRelation}.

\section{Matrix-product state encoding of DTFIM steady state}
\label{section:mps}

    In this section we first briefly discuss the exact steady-state solution $\hat \rho_{\rm ss}$ of Eq.~\eqref{eq:tfim-lindbladian} for the DTFIM before explicitly constructing an exact matrix-product representation that allows efficient numerics. The steady state $\hat \rho_{\rm ss}$ was first obtained in Ref.~\cite{roberts2023exact} with a later work~\cite{noachtar_exact_2026} connecting it to a more general construction of hTRS solutions. As we discussed in App.~\ref{appendix:htrs} we can obtain $\hat \rho_{\rm ss}$ through its purification $\ket\Psi$ in the hTRS framework. In Ref.~\cite{noachtar_exact_2026} the pure, \emph{unnormalized} steady-state purification $\ket {\tilde \Psi(\Omega_j)} \propto \ket{\Psi(\Omega_j)}$ of this doubled system was constructed in the form
    \begin{align}
        \ket {\tilde \Psi(\Omega_j)} = \sum_{k=0}^N \psi_k 
        \left[ \binom N k \right]^{-1/2} \ket{D^N_k}
        \label{eq:pure-steady-state}
    \end{align}
    where $\ket{D^N_k}$ are a set of modified Dicke states given by
    \begin{align}
        \ket{D^N_k}  = \sum_{\substack{s_j\in \{0,1\} \\\sum_j {s_j} = k}} \prod_j^N \alpha_j^{s_j} \ket{\tilde s_j} 
    \end{align}
    where $\alpha_j=\Omega/\Omega_j$ and $\Omega\neq 0$ is a reference value that the final state does \emph{not} depend on. We express these modified Dicke states in terms of two basis states $\ket{\tilde s_j}$ for each spin pair which we label by $s_j\in\{0,1\}$ for convenience. These basis states are explicitly given in terms of the spin states ($\ket{\downarrow}, \ket\uparrow$)
    \begin{align}
    \label{eq:mps-spin-states}
        \ket{\tilde 0_j} &= \frac{1}{\sqrt 2}(\ket{\uparrow\downarrow_j}_{AB} + \ket{\downarrow\uparrow_j}_{AB})\\
        \ket{\tilde 1_j} &= \ket{\downarrow\downarrow}_{AB}.
    \end{align}
    It will be convenient to express the product over these spin states over all the sites as \(\ket{\tilde{\vec s}}\equiv\bigotimes_j\ket{\tilde s_j}\), yielding
    \begin{equation}
       \ket{D^N_k} =  \sum_{\vec{s}, |\vec{s}| = k} \left(\prod_j^N \alpha_j^{s_j} \right)\ket{\tilde{\vec s}}
       \label{eq:modified-dicke-states}
    \end{equation}
    with $|\vec s| = \sum_j s_j$. The complex coefficients $\psi_k \in \mathbb C$ can be obtained from the recursion relation
    \begin{align}
        \psi_{k+1} = \sqrt{\frac{2(N-k)}{k+1}}\frac{J k/N - \Delta + i \kappa / 2}{\Omega} \psi_k
    \end{align}
    where we emphasize that $\Omega\neq 0$ is a real reference parameter since $\ket{\tilde\Psi(\Omega_j)}$ only depends on the set of $\Omega_j$, but not $\Omega$. Note that this is the same as defining the state $\ket{\tilde \Psi(\Omega_j)}$ with respect to a reference state with uniform drives $\Omega_j = \Omega$ in Theorem~\ref{theorem:tfim}. We can then evaluate arbitrary observables $\hat O$ in the single system by calculating
    \begin{equation}
    \label{eq:mps-ss-obs}
        \langle \hat O\rangle_{ss} = \frac{\bra{\tilde \Psi(\Omega_j)} \hat O \otimes \hat{\mathbb I}\ket{\tilde\Psi(\Omega_j)}}{\braket{\tilde \Psi(\Omega_j)}{\tilde \Psi(\Omega_j)}}.
    \end{equation}

   We now explicitly construct an \emph{exact} matrix-product state (MPS) representation of $\ket{\tilde \Psi(\Omega_j)}$ which allows us to obtain these steady-state expectation values efficiently for the DTFIM. 
    
    First, let us quickly review the basics of matrix-product states; for an extensive review we refer to~\cite{schollwock_density-matrix_2011}. Let $\ket \psi$ be a state in the tensor product of $N$ local Hilbert spaces---for simplicity we assume that each local Hilbert space corresponds to a two-level system. In general we can write
    \begin{align}
        \ket \psi = \sum_{s_1,\dots, s_N\in\{0,1\}^N} c_{s_1, \dots, s_N}\ket{s_1}\otimes \dots\otimes \ket{s_N}
        \label{eq:general-quantum-state}
    \end{align}
    with $2^N$ coefficients $c_{s_1, \dots, s_N}\in \mathbb C$. An MPS encodes these coefficients as a product of complex matrices $A^{[s_i]}$ for each local Hilbert space such that
    \begin{align}
        c_{s_1, \dots, s_N} = \sum_{\alpha_1, \dots, \alpha_N=1}^\chi v_{L, \alpha_1}^\dagger A^{[s_1]}_{\alpha_1, \alpha_2} A^{[s_2]}_{\alpha_2, \alpha_3}\dots A^{[s_{N-1}]}_{\alpha_{N-1}, \alpha_N} A^{[s_N]}_{\alpha_N, \alpha_{N+1}} v_{R,\alpha_{N+1}} = v_L^\dagger A^{[s_1]}\dots A^{[s_N]} v_R
        \label{eq:mps-state-coeffs}
    \end{align}
    where $\chi$ is referred to as the bond dimension and $v_L\in \mathbb C^{\chi}, v_R\in \mathbb C^\chi$ as the left and right boundary vectors which effectively correspond to boundary conditions that get absorbed into $A^{[s_1]} \to v_L^\dagger A^{[s_1]}$ and $A^{[s_N]}\to A^{[s_N]} v_R$. The second equality uses a simplified notation in terms of eponymous matrix products. The goal of an MPS is finding the $N$ matrices $A^{[s_i]}$ which leaves us with $\mathcal O(N\times \chi^2)$ parameters instead of the $2^N$ coefficients $c_{s_1, \dots, s_N}$. The matrices can be obtained from the full state coefficients $c_{s_1,\dots,s_N}$ by performing successive SVDs as described in Ref.~\cite{schollwock_density-matrix_2011}, but are not unique which can be exploited to optimize the bond dimension of the MPS.

    MPS are particularly well suited to one-dimensional geometries due to their product structure. Since our system lacks any spatial structure, it may not be intuitive that the steady-state purification $\ket{\tilde \Psi(\Omega_j)}$ possesses a simple MPS representation. However, Ref.~\cite{noachtar_exact_2026} established that the steady-state purification has a simple, one-dimensional structure in an effective charge space. In this charge space the state $|\tilde 1_j\rangle$ corresponds to the vacuum state on site $j$, carrying charge 0, and the $|\tilde 0_j\rangle$ state represents a charge on site $j$. We leverage the one-dimensional structure of this charge space to construct an exact MPS representation of $\ket{\tilde \Psi(\Omega_j)}$.
    
    Let us think of the virtual index $\alpha_j$ as tracking a charge coordinate such as a number of charges. Comparing with Eq.~\eqref{eq:pure-steady-state} it corresponds to the index $k$ which counts the number of charges in the modified Dicke state $|D_k^N\rangle$. Since each site $j$ can hold either 0 or 1 charge, we have to keep track of at most $N$ charges, i.e.~$\alpha_j=0, \dots, N\; \forall j$. Therefore the MPS has at most bond dimension $\chi \leq N + 1$ for any choice of drives $\{\Omega_j\}$. The matrix $A^{[s_j]}$ acts on the charge space and describes how the charge coordinate $k$ is affected if a single charge is added on site $j$. The boundary vectors $v_R$ and $v_L^\dagger$ essentially encode an initial and final state on the charge space.
    
    The difficult aspect of constructing an \emph{exact} MPS encoding of Eq.~\eqref{eq:pure-steady-state} is the dependence on the local parameters $\Omega_j$. For an arbitrary choice of $\Omega_j$ the state lacks any spatial structure and translational invariance which are usually required to find an exact MPS description. Crucially however, Theorem~\ref{theorem:tfim} tells us something much stronger: a purification $\ket{\Psi(\Omega_j)}$ for a DTFIM with some drives \(\{\Omega_j \}\) is related to a DTFIM  with drives in a fully permutationally invariant limit $\Omega_j=\Omega \; \forall j$ (whose steady state purification we refer to as $\ket{\Psi_{pi}} = \ket{\tilde \Psi(\Omega_j=\Omega)}$), via a local similarity transform $\hat V$:
    \begin{align}
        \ket{\tilde \Psi(\Omega_j)} = \underbrace{e^{\sum_j \lambda_j \hat Z_j}}_{\equiv \hat V} \ket{\Psi_{pi}} = \bigotimes_j e^{\lambda_j \hat Z_j} \ket{\Psi_{pi}}.
        \label{eq:dtfim-states-similarity-transform}
    \end{align}
    Here, the parameters $\lambda_j$ are given by
    \begin{align}
        \lambda_j = \log(\frac{\Omega_j}{\Omega})
    \end{align}
    with $\hat Z_j = (\hat \sigma^j_{z,A} + \hat \sigma^{j}_{z,B}) / 2$. This local similarity transform is the reason for the efficient MPS encoding of $\ket{\tilde \Psi(\Omega_j)}$ despite its lack of spatial structure: the permutationally invariant state can be naturally encoded as a particularly simple MPS due to its decomposition into Dicke states~\cite{florido-llinas_product_2025,raveh_dicke_2024} as we will see. The local similarity transform simply corresponds to rescaling the MPS matrices $A^{[s_j]}\to (\Omega_j / \Omega)^{-s_j} A^{[s_j]}$ of the permutationally invariant state. This is intuitive since the action of $\exp(\lambda_j \hat Z_j)$ on the two local basis states is
    \begin{align}
        e^{\lambda _j \hat Z_j} \ket{\tilde 0_j} &= \ket{\tilde 0_j}\\
        e^{\lambda _j \hat Z_j} \ket{\tilde 1_j} &= \frac{\Omega}{\Omega_j}\ket{\tilde 1_j}.
    \end{align}
    Therefore the local similarity transform preserves the bond dimension of the MPS. Since the MPS encoding is exact, any ordering of $\Omega_j$ gives the same MPS up to trivial permutations.
    Note that the locality of $\hat V$, i.e. as a product over local transforms, does not imply that the inhomogeneity in $\Omega_j$ only affects the state $\ket{\tilde \Psi(\Omega_j)}$ locally (unlike a unitary transform). To understand the intuition behind this, let us consider how $\hat V$ acts on a product state in the $\{\ket{\tilde 0}, \ket{\tilde 1}\}^{\otimes N}$ basis: it only changes the state's norm by a state-dependent amount. Therefore if we act with $\hat V$ on a superposition of product states, the relative weights of the superposition will change which results in a nontrivial global change of the state. This effect can be seen in Eq.~\eqref{eq:modified-dicke-states}. As a consequence any reduced state of $\ket{\tilde \Psi(\Omega_j)}$ nontrivially depends on all $\Omega_j$. However, so far, only the fully permutationally invariant limit ($\Omega_j=\Omega \; \forall j$) was studied in detail~\cite{roberts2023exact,noachtar_exact_2026}. In this limit $\ket{D^N_k}$ exactly takes the form of a Dicke state
    \begin{align}
        \lim_{\Omega_j \to \Omega}\ket{D^N_k} = \sum_{\text{permutations }\mathcal P} \mathcal P\ket{\tilde 0_j}^{\otimes N-k}\ket{\tilde 1_j}^{\otimes k}
    \end{align}
    which are permutationally invariant by construction. This permutational invariance will be a crucial factor in our MPS representation.

    First of all, let us notice that for a permutationally invariant state we can choose the matrices $A^{[s_i]}$ to be independent of $i$, i.e. $A^{[s_i]}\to A^{[s]}$, and, more importantly, to commute, i.e. $[A^{[s_i]}, A^{[s_j]}] =0 \; \forall s_i, s_j$ since swapping $i$ and $i+1$ must give the same coefficient in Eq.~\eqref{eq:general-quantum-state}:
    \begin{align}
        v_L^\dagger A^{[s_1]}\dots A^{[s_i]} A^{[s_{i+1}]} \dots A^{[s_N]} v_R= c_{s_1, \dots, s_i, s_{i+1},\dots,s_N} = c_{s_1, \dots, s_{i+1}, s_{i},\dots,s_N} = v_L^\dagger A^{[s_1]}\dots A^{[s_{i+1}]} A^{[s_{i}]} \dots A^{[s_N]} v_R.
    \end{align}
    Since the state $\ket{\Psi_{pi}}$ only has two basis states for the Hilbert space of each spin pair, this means that we can simply pick $A^{[0]}$ to be proportional to the identity $A^{[0]}_{\alpha, \beta}\propto \delta_{\alpha,\beta}$. Furthermore for any bit string $\vec s\in \{0,1\}^N$ with Manhattan norm $\sum_i s_i = k$ the $A^{[1]}$ matrix must satisfy
    \begin{align}
        c_{\vec s} = v_L^\dagger \left ( A^{[1]}\right )^k v_R \equiv C_k,
    \end{align}
    i.e. states with a fixed number $k$ of $\ket{\tilde 1}$ states have the same weight $C_k$. Comparing this with Eq.~\eqref{eq:pure-steady-state} in the permutationally invariant limit $\Omega_j\to \Omega$ we identify $C_k=\psi_k/\sqrt{\binom{N}{k}}$. One choice of $v_L^\dagger, v_R$ and $A^{[1]}$ that recovers $\ket{\Psi_{pi}}$ is~\cite{raveh_dicke_2024}
    \begin{align}
        v_L^\dagger = \left (\psi_0 \dots \frac{\psi_k}{\sqrt{\binom N k}} \dots \psi_N\right ), v_R = \begin{pmatrix}1 \\ 0 \\ \vdots\end{pmatrix}, A^{[1]}_{\alpha, \beta} = \delta_{\alpha, \beta + 1},
    \end{align}
    where $v_L$ and $v_R$ are length $N+1$ and $A^{[1]}$ is an $N+1\times N+1$ matrix with ones along the first lower offdiagonal; consequently $A^{[0]}$ is the $N+1\times N+1$ identity matrix.
    This structure has a simple interpretation in the charge space we discussed earlier (also see Ref.~\cite{noachtar_exact_2026}): adding a charge on site $j$ increases the charge coordinate (virtual index) by exactly one, while not adding a particle leaves the charge coordinate unchanged. Since the initial state $v_R$ essentially represents a vacuum state with no charges, we can think of the MPS as a sequential state automaton that successively builds the steady-state purification starting from $\ket{D^N_0}$~\cite{gopalakrishnan_push-down_2025}. The final state $v_L^\dagger$ ensures the correct relative weights of $\ket{\Psi_{pi}}$.
    In this basis the MPS of $\ket{\Psi_{pi}}$ has an exact bond dimension $\chi=N+1$. Since our state of interest $\ket{\tilde \Psi(\Omega_j)}$ is related to $\ket{\Psi_{pi}}$ via Eq.~\eqref{eq:dtfim-states-similarity-transform}, we obtain the MPS for $\ket{\tilde \Psi(\Omega_j)}$ simply by taking
    \begin{align}
        A^{[0_j]} \to A^{[0_j]} \text{ and } A^{[1_j]} \to \frac{\Omega}{\Omega_j}A^{[1_j]}
    \end{align}
    which highlights that the bond dimension is unchanged by the similarity transform. 
    The ordering of the drives $\Omega_j$ onto the MPS sites does not affect the properties of our MPS and can be arbitrarily chosen; in any case the maximum bond dimension is $N+1$ for any ordering. 
    A recent work~\cite{florido-llinas_product_2025} has investigated a similar class of states that permit an MPS encoding with a constant bond dimension under any permutation of sites. Since we have shown that the MPS representation of the steady-state purification constructed here has an extensive bond dimension, our state does not necessarily meet the requirements to be approximated by a finite sum of product states. Therefore we consider our state to have a non-trivial structure and the MPS encoding to be an effective strategy for performing steady-state numerics.

    \subsection*{Numerical implementation}
    The MPS representation of the steady-state purification $\ket{\tilde \Psi(\Omega_j)}$ is formally exact. Therefore we can implement it without relying on any approximations to obtain local expectation values using Eq.~\eqref{eq:mps-ss-obs}. For the results in the main text we have used the TeNPy library~\cite{hauschild_efficient_2018} for steady-state computations. Note that since the MPS representation is exact, we are only limited by machine precision, even for hundreds of spins due to the polynomial scaling of the network in $N$. No other approximations or assumptions are required.

\section{Matrix-product state encoding of bosonic lattice model}
\label{section:bosonic-mps}

    Here we introduce the exact steady-state solution \(\hat \rho_{\rm ss}\) for the bosonic model in Eq.~\eqref{eq:bosonic-lind} by constructing its purification \(\ket{\Phi}\). We use the exact form of $\ket\Phi$ to construct an MPS encoding of the state which enables efficient numerics for steady-state expectation values similarly to App.~\ref{section:mps}. In Ref.~\cite{roberts2023competition} the steady-state purification $\ket\Phi$ was shown to take the form
    \begin{align}
        \ket \Phi = \sum_{m=0}^\infty \frac{\phi_m}{m!} (\hat K_+)^m \ket{vac} \text{ such that }  \hat \rho_{\rm ss} = \Tr_B{\hat \rho_{AB}} \text{ with } \hat \rho_{AB} = \ketbra{\Phi}{\Phi}.
        \label{eq:bosonic-steady-state-purification}
    \end{align}
    where \(\ket{vac}\) is the vacuum state of all modes, i.e. \(\hat a_{i,A} \ket{vac}=0 = \hat a_{i,B} \ket{vac}\) for all \(i\). The operator \(\hat K_+\) essentially pairs particles into the modes \(\hat \alpha_i = (\hat a_{i,A} + \hat a_{i,B})/\sqrt 2\) and takes the form
    \begin{align}
        \hat K_+ = \sum_{i,j} M_{ij} \hat \alpha_i^\dagger \hat \alpha_j^\dagger.
    \end{align}
    Here, we again assume \(M\) to be a real symmetric matrix. The state coefficients \(\phi_m\) are given by the recursion relation
    \begin{align}
    \label{eq:recursion-relation-bosonic}
        \phi_m = \frac{N}{2U}\frac{\phi_{m-1}}{\delta + m - 1}
    \end{align}
    where \(\delta = 1 - N (\Delta + i\kappa / 2)/ (2U)\). For a thorough analysis of the steady-state phase diagram we refer to Ref.~\cite{roberts2023competition}.

    We are now set up to construct an MPS encoding for the steady-state purification $\ket\Phi$ of this model.
    The MPS encoding for the purification \(\ket \Phi\) of the steady state should allow us to run scalable numerics for \emph{any} real-symmetric matrix \(M\). Surprisingly, our construction will be completely analogous to the spin case in many ways: we will use Theorem~\ref{theorem:bosonic} to map \(\ket \Phi\) to a permutationally invariant state via a local similarity transform which preserves the bond dimension. This permutationally invariant state will then have a simple MPS encoding.

    First we note that it is convenient to work in a basis that diagonalizes the pairing $M_{ij}$ matrix (see App.~\ref{app:bosonic-proof-final}) which we assume to be real and symmetric~\footnote{It is straightforward to extend to complex symmetric matrices via Takagi factorization.}. Let us take
    \begin{align}
        M_{ij} = \sum_p u_{ip} d_p u_{jp}
    \end{align}
    where \(d_p \in \mathbb R\) and \(u_{ip} \in \mathbb R\) are the elements of a real, orthogonal matrix, i.e. \(\sum_p u_{ip} u_{jp} = \delta_{ij}\). This allows us to write 
    \begin{align}
        \hat K_+ = \sum_{p} d_p \Big(\sum_i u_{ip} \hat\alpha^\dagger _i\Big ) \Big(\sum_j u_{jp} \hat\alpha^\dagger _j\Big ) = \sum_p d_p (\hat u_p^\dagger)^2
        \label{eq:diag-pair-creation-op}
    \end{align}
    where we have defined the diagonal modes
    \begin{align}
        \hat{u}_{p}^{\dag} = \sum_i u_{ip} \hat{\alpha}_i^{\dag}.
    \end{align}
    This has the advantage that $\hat K_+$ simply creates a photon pair in each mode $\hat u_p$ weighted by the driving strength $d_p$ in Eq.~\eqref{eq:bosonic-steady-state-purification}. The state $(\hat K_+)^m \ket{vac}$ corresponds to $m$ photon pairs delocalized across $N$ modes and weighted by their relative driving strengths $d_p$; we can think of this as the bosonic parallel of the modified Dicke states $\ket{D^N_k}$ of the DTFIM. Therefore, the structure of this state is remarkably similar to the structure of the DTFIM steady-state purification in App.~\ref{section:mps} which allows us to use similar ideas to construct an MPS encoding.
    
    We will now construct the MPS in the Fock basis of the $\hat u_p$ modes; we will later show that expectation values of the ``real-space" $\hat a_i$ modes can be recovered by taking linear combinations of $\hat u_p$ expectation values. This has one crucial advantage: Theorem~\ref{theorem:bosonic} tells us that steady-state purification of the $\hat u_p$ modes (cf. Eq.~\eqref{eq:diag-pair-creation-op}) is related to a permutationally invariant state $\ket{\Phi'}$ via a local similarity transform $\hat\Lambda$; we will show that $\ket{\Phi'}$ has a simple and exact MPS encoding. According to Theorem~\ref{theorem:bosonic} this local similarity transform is
    \begin{align}
        \hat \Lambda &= e^{\sum_{p,q}\lambda_{pq}\hat Z_{pq}}
    \end{align}
    with the coefficients $\lambda_{pq}\in\mathbb C$ and operator $\hat Z_{pq}$
    \begin{align}
        \lambda_{pq} &= \delta_{pq} \log(\frac{d_p}{\bar d})\\
        \hat Z_{pq} &= \frac{1}{2}\Big (\hat a^\dagger_{p,A} \hat a_{q,A} + \hat a^\dagger_{p,B} \hat a_{q,B} \Big )\\
        &= \frac{1}{4}\Big[\big(\hat{u}_p^{\dagger}\hat{u}_q + \hat{u}_q^{\dagger}\hat{u}_p\big) + \big(\hat{v}_p^{\dagger}\hat{v}_q + \hat{v}_q^{\dagger}\hat{v}_p\big) + \big(\hat{u}_p^{\dagger}\hat{v}_q - \hat{v}_q^{\dagger}\hat{u}_p\big) + \big(\hat{v}_p^{\dagger}\hat{u}_q - \hat{u}_q^{\dagger}\hat{v}_p\big)\Big].
    \end{align}
    Here, $\bar d >0$ is simply a reference value like discussed in App.~\ref{app:bosonic-diagonal}. In the final line we also introduced the $\hat v_p$ modes
    \begin{align}
        \hat{v}_p^{\dagger} = \sum_i u_{ip}\frac{\hat a_{i,A}^\dagger - \hat a_{i,B}^\dagger}{\sqrt 2}
    \end{align}
    which are defined with respect to the odd-parity combination of the $\hat a_{i,A}$ and $\hat a_{i,B}$ mode. This operator fulfills canonical commutation relations with $\hat u_p$, i.e. $[\hat v_p, \hat u_q] = 0 = [\hat v_p, \hat u^\dagger_q]$ and $[\hat v_p, \hat v_q]=0, [\hat v_p, \hat v_q^\dagger] =\delta_{p,q}$. This means that
    \begin{align}
        \hat \Lambda = e^{\sum_p \frac{\lambda_{pp}}{2} (\hat u^\dagger_p \hat u_p + \hat v_p^\dagger \hat v_p)} = e^{\sum_p \frac{\lambda_{pp}}{2} \hat u^\dagger_p \hat u_p} e^{\sum_p \frac{\lambda_{pp}}{2} \hat v^\dagger_p \hat v_p}.
        \label{eq:diag-modes-similarity-transform}
    \end{align}
    Since the steady-state purification is in the vacuum with respect to the odd-parity mode, i.e. $(\hat a_{i,A} - \hat a_{i,B})\ket{\Phi}=0$, we know that $\hat v_p \ket{\Phi}=0$ which means that the second exponential in the final line acts trivially.
    
    From here on, our approach closely matches the MPS construction of the DTFIM in Sec.~\ref{section:mps}. Since the similarity transform \(\hat \Lambda\) is local in the \(\hat u_p\) modes, we know that if \(\ket{\Phi'}\) has an efficient MPS encoding in terms of the \(\hat u_p\) modes, then so does \(\ket \Phi\). We therefore just need to find an MPS encoding of the permutationally invariant state $\ket{\Phi'}$ with uniform driving amplitudes $d_p = \bar d\; \forall p$.

    We will now construct an efficient MPS encoding of the state 
    \begin{align}
        \ket{\Phi'} = \sum_{n_1, \dots, n_N \in \{0,1,2,\dots,n_c\}} c_{n_1, \dots,n_N} \ket{2n_1,0}\otimes \dots \otimes \ket{2n_N,0}
        \label{eq:collective-bosonic-state}
    \end{align}
    where \(\ket{2n_p,0}\propto (\hat u_p^\dagger)^{2n_p} \ket{vac}\); remember that each mode \(\hat u_p\) has the partner mode $\hat v_p$ in the doubled Hilbert space which remains in the vacuum state. Also note that we explicitly write a photon cutoff \(n_c\) in Eq.~\eqref{eq:collective-bosonic-state}. To \emph{exactly} represent the state we will need \(n_c=\infty\), but in practice we will limit this to a finite value at which expectation values have converged upon increasing $n_c$. The state coefficients \(c_{n_1, \dots, n_N}\) are given in the form of Eq.~\eqref{eq:mps-state-coeffs} as a matrix product. We now have to find the corresponding matrices \(A^{[n]}\) and boundary vectors \(v_{L}^\dagger\) and \(v_R\) for the state \(\ket{\Phi'}\) for a specified bond dimension \(\chi\) (which will be unbounded to exactly represent the state). 
    
    We again consider the effective charge space discussed in Sec.~\ref{section:mps} for the DTFIM to think of our system as one-dimensional. Each photon pair represents a single charge. Basically we arrange the diagonal modes $\hat u_p$ on a one-dimensional lattice and introduce a matrix $A^{[n_p]}_{\alpha_p, \alpha_{p+1}}$ for each Fock state with $n_p$ photon pairs in the $\hat u_p$ mode. The virtual index $\alpha_p$ corresponds to the charge coordinate. The matrix $A^{[n_p]}$ describes how the charge coordinate is affected by adding $n_p$ photon pairs on site $p$. Since the system is bosonic, more than one photon pair/charge can be injected per site which is different from the DTFIM and makes the state more complex and potentially interesting. 
    
    We now make an ansatz for $v_L^\dagger, v_R$ and $A^{[n]}$: first, we again successively build the state in the charge space starting from the vacuum state. This means that the ``initial state" is $v_{R,\alpha} = \delta_{\alpha,0}$. Since we want the matrix $A^{[n_p]}$ to increase the charge coordinate by $n_p$ upon injecting $n_p$ photon pairs on site $p$, we make the simple ansatz $A^{[n]}_{\alpha,\beta} \propto \delta_{\alpha, \beta +n}$. This is compatible with the state $\ket{\Phi'}$ being permutationally invariant, since all $A^{[n]}$ matrices trivially commute. Finally we want all states with $m$ photon pairs distributed across the $N$ modes to be proportional to the coefficient $\phi_m$ such that we choose the ``final state" $v_{L,m}^\dagger \propto \phi_m$. By expanding the state $\ket{\Phi'}$ in Eq.~\eqref{eq:collective-bosonic-state} in terms of Fock states it can be confirmed that the ansatz
    \begin{align}
        v_L^\dagger = (\phi_0\dots \phi_k \dots \phi_\chi), v_R = \begin{pmatrix}1\\0\\\vdots\end{pmatrix}, A^{[n]}_{\alpha,\beta} = \bar d^n \sqrt{\binom{2n}{n}}\delta_{\alpha, \beta + n}
        \label{eq:bosonic-mps-tensors}
    \end{align}
    indeed recovers the exact steady-state purification in Eq.~\eqref{eq:bosonic-steady-state-purification} with $M_{ij} = \bar d \,\delta_{i,j}$, i.e.~in the permutationally invariant limit of $\ket{\Phi'}$. The factor \(\sqrt{\binom{2n}{n}}\) is a consequence of the multinomial expansion for the uniform pair creation operator 
    \begin{align}
        \hat K_+' = \bar d \sum_p \left (\hat u_p^\dagger \right )^2
    \end{align}
    and the normalization of \((\hat u_p^\dagger )^{2n_p} \ket{vac}\) which ensures the correct relative weights between all states; the factor \(m!\) in Eq.~\eqref{eq:bosonic-steady-state-purification} is also canceled as a result of the multinomial expansion. It is then straightforward to see that for any Fock product state with the photons distributed according to \(\{n_p\}\)
    \begin{align}
        v_L^\dagger \Big [\prod_k A^{[n_p]} \Big ] v_R \propto \phi_{\sum_p n_p}
    \end{align}
    maps to the coefficient corresponding to the total number of photons \(\sum_p n_p\). Finally we highlight that an exact MPS representation formally requires $n_c=\infty$ and therefore also $\chi = \infty$. If we truncate $n_c<\infty$, we only have an approximate representation of the steady-state purification. For any $n_c<\infty$ the bond dimension is at most $\chi = N n_c + 1$ which is related to the maximum number of charges that the system can hold; for $n_c = 1$ this exactly recovers the bond dimension of Sec.~\ref{section:mps}.

    To go from the MPS for \(\ket{\Phi'}\) to the MPS for \(\ket \Phi\) we simply replace
    \begin{align}
        A^{[n_p]} \to \left ( \frac{d_p}{\bar d} \right )^{n_p} A^{[n_p]}
    \end{align}
    which is a consequence of the form of Eq.~\eqref{eq:diag-modes-similarity-transform}. Therefore we have constructed an \emph{exact} MPS of the $\ket\Phi$ state that allows us to evaluate local expectation values in the \(\hat u_p\) modes according to Eq.~\eqref{eq:mps-ss-obs}
    using resources that scale polynomially in the total number of modes $N$, the photon cutoff $n_{c}$ and the bond dimension $\chi$. 
    To obtain expectation values of the ``real-space" modes \(\hat a_i\) such as \(\langle \hat a_i \hat a_j\rangle_{ss}\) we simply decompose
    \begin{align}
        \langle \hat a_i \hat a_j\rangle_{ss} = \langle \hat a_{i,A} \hat a_{j,A}\rangle_{ss} = \frac{1}{2}\sum_{k,\ell} u_{ik} u_{j\ell} \langle \hat u_{k}\hat u_\ell\rangle_{ss}
    \end{align}
    where we used the simple fact that $(\hat a_{i,A} - \hat a_{i,B})\ket \Phi/\sqrt 2=0$. This dramatically simplifies since the expectation values $\ev{\hat u_k \hat u_\ell}$ are constrained by
    \begin{align}
        \langle 2n_k,0|\langle 2n_\ell,0| \hat u_k \hat u_\ell |2n_k',0\rangle |2n_\ell',0\rangle \propto \delta_{k,\ell}.
    \end{align}
    This implies that \(\langle \hat u_{k}\hat u_\ell\rangle_{ss}\propto \delta_{k,\ell}\) which is a result of a \(\mathbb Z_2\) symmetry of the steady-state purification (or a weak symmetry of our model under $\hat a_i \to -\hat a_i$). Therefore, any observable that is quadratic in \(\hat a_{i,A}, \hat a_{i,A}^\dagger, \hat a_{i,B}, \hat a_{i,B}^\dagger\), no matter how non-local, can be evaluated in terms of \emph{local} expectation values of \(\hat u_k, \hat u_k^\dagger\) up to some postprocessing of the \(u_{ik}\) matrix elements which is efficient. 

    \subsection*{Numerical implementation}
    The MPS representation for the steady-state purification $\ket\Phi$ constructed in the previous section is formally exact, i.e.~it requires no approximation. To obtain the desired expectation values we implement the MPS in TeNPy~\cite{hauschild_efficient_2018} (analogous to the DTFIM). Naturally this requires us to introduce a local photon number cutoff $n_c<\infty$ (recall that the exact MPS requires $n_c=\infty$). Since a finite $n_c$ also implies a maximum bond dimension $\chi = \mathcal O(N n_c)$, we effectively have one lever to control convergence to exact expectation values (as long as $\chi = N n_c + 1$).  
    
    We now make a heuristic argument to explain why a finite-$n_c$ approximation is controlled. Let us consider the amplitude $c_{n_1,\dots,n_N}$ of a Fock state with photon pair distribution $\{n_p\}$. Simply by contracting the MPS we find that
    \begin{align}
        |c_{n_1, \dots, n_N}| &\sim |\phi_{\sum_p n_p} \prod_p \left [ d_p^{n_p} \sqrt{\binom{2n_p}{n_p}} \right ]| \sim \prod_{j=1}^{\sum_p n_p}\left | \frac{N}{j - N \frac{\Delta + i \kappa/2}{2U}} \right| \prod_p \left| \frac{d_p}{2U}\right|^{n_p}\sqrt{\binom{2n_p}{n_p}} \\
        &= \exp\left [ - \sum_{j=1}^{\sum_p n_p} \log \left|\frac{j}{N} - \frac{\Delta + i \kappa/2}{2U} \right| + \sum_p n_p \log\left |\frac{d_p}{2U}\right | + \frac{1}{2}\sum_p \log \binom{2n_p}{n_p} \right ]\\
        &\leq \exp\left [ - \sum_{j=1}^{\sum_p n_p} \log \left|\frac{j}{N} - \frac{\Delta + i \kappa/2}{2U} \right| + \sum_p n_p \log \left |\frac{d_p}{U}\right| \right] \\
        &\leq \exp\left [ - \sum_{j=1}^{\sum_p n_p} \log \left|\frac{j}{N} - \frac{\Delta + i \kappa/2}{2U} \right| + \log\frac{d_{max}}{|U|}\sum_p n_p \right]\\
        &= \exp\left [ \sum_{j=1}^{\sum_p n_p} \left ( \log \frac{d_{max}}{|U|} - \log\left|\frac{j}{N} - \frac{\Delta + i \kappa/2}{2U} \right| \right )\right ]\\
        &\sim \exp\left [ -\sum_p n_p \left ( \int_0^1 dx \log \Big |\frac{\sum_p n_p}{N}  x - \frac{\Delta + i \kappa / 2}{2U} \Big | - \log \frac{d_{max}}{|U|}\right ) \right]
    \end{align}
    where we defined $d_{max} = \max_p |d_p|$ and we used that $\binom{2n_p}{n_p} \leq 4^{n_p}$. Since the second term inside the sum becomes increasingly negative for large $j$ while the first term stays constant, it immediately follows that the term inside the exponential becomes negative for sufficiently large $\sum_p n_p$. Therefore the coefficient eventually decays at least exponentially with $\sum_p n_p$. We argue that this allows us to justify a finite photon cutoff $n_c$ since contributions from high-photon number Fock states decay rapidly. Since evaluating any expectation values scales at worst polynomially in $N, n_c$ and $\chi$ (if we take $\chi$ to be independent of $n_c$ and $N$), numerical computations are tractable for many bosonic models. 

For the purposes of the simulations for Fig.~\ref{fig:bosonic} in the main text, we use a photon cutoff of $n_c =10$, with $\chi = Nn_c + 1 = 41$. In Fig.~\ref{fig:bosonic-convergence}, we plot how some of the relevant observables scale in $n_c$ for the same system and parameters used in Fig.~\ref{fig:bosonic}. This serves as a useful heuristic for convergence. 

\begin{figure}
    \centering
    \includegraphics[width=0.5\linewidth]{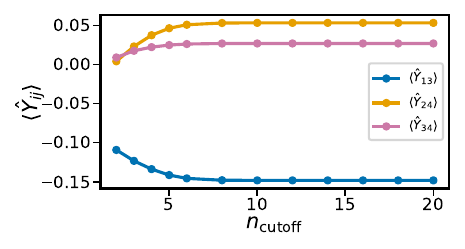}
    \caption{Output observables for the $N= 4$ bosonic model with parameters corresponding to the teacher system (for $x_{12}/U = 0.5$). For each photon cutoff, we use a bond dimension $\chi = N n_c + 1$.}
    \label{fig:bosonic-convergence}
\end{figure}

\section{Additional training details}
\label{section:dtfim-teacher-student-examples}

\subsection{Additional DTFIM training example}

In Fig.~\ref{fig:teacher-student-app}, we present an additional training example of teacher-student learning for a system similar to the one presented in Fig.~\ref{fig:teacher-student-compare}.
It is for an \(N = 11\) system, with the key difference being that the output is now ten dimensional; there are no internal spins.
In Fig.~\ref{fig:teacher-student-app}(a) we plot the teacher and learned student outputs for all ten output spins.
Notice now that all the student drives match the teacher drives almost exactly (see Table~\ref{table:dtfim-2}) - since each of them is associated with an output spin.

\begin{figure}
    \centering
    \includegraphics[width=0.5\linewidth]{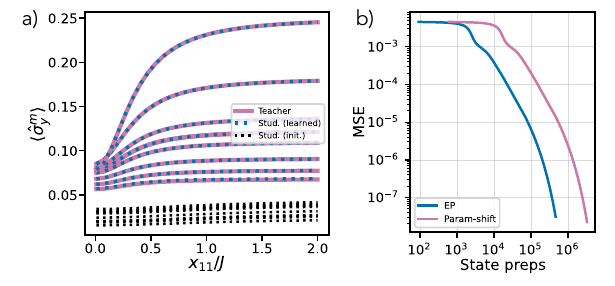}
    \caption{ (a) DTFIM outputs (\(\langle \hat{\sigma}_y^1\rangle, \dots,\langle \hat{\sigma}_y^{10}\rangle\)) for a \(N = 11\) spin system with \(J/\kappa\) = 5, trained using both EP and parameter shift. We plot both the teacher outputs (pink) along with initial (black dotted) and trained (blue) student outputs. Values for teacher and student drives (before and after training) are included in Table~\ref{table:dtfim-2} in App.~\ref{section:dtfim-teacher-student-examples}. (b) The loss plotted over number of experiments per training sample for both training techniques. An experiment is counted each time a new steady state needs to be prepared.}
    \label{fig:teacher-student-app}
\end{figure}

\subsection{Training parameters}
We include some tables containing training values for the teacher student examples throughout the text.

\begin{table}[h]
\centering
\begin{tabular}{ccccc}
\hline
   &   &   & $\theta_i/J$ &   \\
\hline
$i$ & Teacher & Init. student & Trained student (EP) & Trained student (param-shift) \\
\hline
$1$  & $+0.200$ & $+0.567$ & $+0.212$ & $+0.212$ \\
$2$  & $+0.400$ & $+0.292$ & $+0.405$ & $+0.405$ \\
$3$  & $+0.600$ & $+0.371$ & $+0.603$ & $+0.603$ \\
$4$  & $+0.800$ & $+0.508$ & $+0.900$ & $+0.901$ \\
$5$  & $+1.000$ & $+0.430$ & $+0.867$ & $+0.868$ \\
$6$  & $+1.200$ & $+0.366$ & $+0.843$ & $+0.843$ \\
$7$  & $+1.400$ & $+0.292$ & $+0.818$ & $+0.818$ \\
$8$  & $+0.700$ & $+0.445$ & $+0.874$ & $+0.874$ \\
$9$  & $+0.700$ & $+0.224$ & $+0.797$ & $+0.797$ \\
$10$ & $+0.700$ & $+0.551$ & $+0.920$ & $+0.921$ \\
$11$ & \multicolumn{4}{c}{input drive (swept)} \\
\hline
\end{tabular}
\caption{Table containing drives $\theta_{j}/J$ for the $N=11$ spin teacher–student run in Fig.~\ref{fig:teacher-student-compare}. We include drives for the teacher system, randomly initialized student, and trained student (for both EP and parameter-shift).}
\label{table:dtfim-1}
\end{table}

\begin{table}[h]
\centering
\begin{tabular}{ccccc}
\hline
   &   &   & $\theta_i/J$ &   \\
\hline
$i$ & Teacher & Init. student & Trained student (EP) & Trained student (param-shift) \\
\hline
$1$  & $+0.200$ & $+0.281$ & $+0.201$ & $+0.201$ \\
$2$  & $+0.400$ & $+0.518$ & $+0.401$ & $+0.400$ \\
$3$  & $+0.600$ & $+0.443$ & $+0.600$ & $+0.600$ \\
$4$  & $+0.800$ & $+0.348$ & $+0.800$ & $+0.800$ \\
$5$  & $+1.000$ & $+0.473$ & $+1.000$ & $+1.000$ \\
$6$  & $+1.200$ & $+0.277$ & $+1.198$ & $+1.198$ \\
$7$  & $+1.400$ & $+0.256$ & $+1.381$ & $+1.384$ \\
$8$  & $+0.700$ & $+0.580$ & $+0.700$ & $+0.700$ \\
$9$  & $+0.700$ & $+0.208$ & $+0.700$ & $+0.700$ \\
$10$ & $+0.700$ & $+0.487$ & $+0.700$ & $+0.700$ \\
$11$ & \multicolumn{4}{c}{input drive (swept)} \\
\hline
\end{tabular}
\caption{Table containing drives $\theta_{j}/J$ for the $N=11$ spin teacher–student run in Fig.~\ref{fig:teacher-student-app}. We include drives for the teacher system, randomly initialized student, and trained student (for both EP and parameter-shift).}
\label{table:dtfim-2}
\end{table}

\begin{table}[h]
\centering
\begin{tabular}{ccccc}
\hline
   &   &   &    $\theta_{ij}/U$ &   \\
\hline
$(i,j)$ & Teacher & Init. student & Trained student (EP) & Trained student (param-shift) \\
\hline
 $(1,1)$ & $+0.058$ & $-0.117$ & $-0.057$ & $-0.010$ \\
 $(1,2)$ & \multicolumn{4}{c}{input drive (swept)} \\
 $(1,3)$ & $+0.365$ & $-0.091$ & $+0.366$ & $+0.370$ \\
$(1,4)$ & $-0.199$ & $+0.270$ & $+0.200$ & $+0.209$ \\
 $(2,2)$ & $+0.364$ & $-0.236$ & $-0.363$ & $-0.322$ \\
 $(2,3)$ & $-0.082$ & $+0.120$ & $+0.082$ & $+0.090$ \\
 $(2,4)$ & $-0.004$ & $+0.210$ & $-0.004$ & $-0.007$ \\
$(3,3)$ & $+0.279$ & $-0.169$ & $-0.278$ & $-0.221$ \\
 $(3,4)$ & $-0.014$ & $-0.261$ & $-0.014$ & $-0.022$ \\
 $(4,4)$ & $+0.085$ & $-0.108$ & $-0.083$ & $-0.049$ \\
\hline
\end{tabular}
\caption{Table containing drives $\theta_{ij}/U$ for the $N=4$ bosonic teacher–student run in Fig.~\ref{fig:bosonic}. We include drives for the teacher system, randomly initialized student, and trained student (for both EP and parameter-shift). We use a photon cutoff $n_c = 10$, and bond dimension $\chi = 41$ for the MPS simulations.}
\label{table:bosonic}
\end{table}
\end{appendix}
\bibliographystyle{ieeetr}
\bibliography{citations}

\begin{thebibliography}{10}

\bibitem{stern2023learning}
M.~Stern and A.~Murugan, ``Learning without neurons in physical systems,'' {\em Annual Review of Condensed Matter Physics}, vol.~14, no.~1, pp.~417--441, 2023.

\bibitem{momeni2025training}
A.~Momeni, B.~Rahmani, B.~Scellier, L.~G. Wright, P.~L. McMahon, C.~C. Wanjura, Y.~Li, A.~Skalli, N.~G. Berloff, T.~Onodera, {\em et~al.}, ``Training of physical neural networks,'' {\em Nature}, vol.~645, no.~8079, pp.~53--61, 2025.

\bibitem{lopez2023self}
V.~Lopez-Pastor and F.~Marquardt, ``Self-learning machines based on hamiltonian echo backpropagation,'' {\em Physical Review X}, vol.~13, no.~3, p.~031020, 2023.

\bibitem{scellier2017equilibrium}
B.~Scellier and Y.~Bengio, ``Equilibrium propagation: Bridging the gap between energy-based models and backpropagation,'' {\em Frontiers in computational neuroscience}, vol.~11, p.~24, 2017.

\bibitem{Scellier2021DeepLearningTheoryPhysics}
B.~Scellier, ``A deep learning theory for neural networks grounded in physics,'' {\em arXiv preprint arXiv:2103.09985}, 2021.

\bibitem{Kendall2020EndToEndAnalogEqProp}
J.~Kendall, R.~Pantone, K.~Manickavasagam, Y.~Bengio, and B.~Scellier, ``Training end-to-end analog neural networks with equilibrium propagation,'' {\em arXiv preprint arXiv:2006.01981}, 2020.

\bibitem{Martin2021EqSpike}
E.~Martin, M.~Ernoult, J.~Laydevant, S.~Li, D.~Querlioz, T.~Petrisor, and J.~Grollier, ``Eqspike: Spike-driven equilibrium propagation for neuromorphic implementations,'' {\em iScience}, vol.~24, no.~3, p.~102222, 2021.

\bibitem{Ernoult2020ContinualWeightUpdates}
M.~Ernoult, J.~Grollier, D.~Querlioz, Y.~Bengio, and B.~Scellier, ``Equilibrium propagation with continual weight updates,'' {\em arXiv preprint arXiv:2005.04168}, 2020.

\bibitem{Scellier2022AgnosticPhysicsDrivenDL}
B.~Scellier, S.~Mishra, Y.~Bengio, and Y.~Ollivier, ``Agnostic physics-driven deep learning,'' {\em arXiv preprint arXiv:2205.15021}, 2022.

\bibitem{Falk2025TemporalContrastiveLearning}
M.~J. Falk, A.~T. Strupp, B.~Scellier, and A.~Murugan, ``Temporal contrastive learning through implicit non-equilibrium memory,'' {\em Nature Communications}, vol.~16, p.~2163, 2025.

\bibitem{Wang2024CoupledPhaseOscillatorsEqProp}
Q.~Wang, C.~C. Wanjura, and F.~Marquardt, ``Training coupled phase oscillators as a neuromorphic platform using equilibrium propagation,'' {\em Neuromorphic Computing and Engineering}, vol.~4, no.~3, p.~034014, 2024.

\bibitem{laydevant2024training}
J.~Laydevant, D.~Markovi{\'c}, and J.~Grollier, ``Training an ising machine with equilibrium propagation,'' {\em Nature Communications}, vol.~15, no.~1, p.~3671, 2024.

\bibitem{Stern2021SupervisedLearningPhysicalNetworks}
M.~Stern, D.~Hexner, J.~W. Rocks, and A.~J. Liu, ``Supervised learning in physical networks: From machine learning to learning machines,'' {\em Physical Review X}, vol.~11, no.~2, p.~021045, 2021.

\bibitem{Stern2024PhysicalEffectsOfLearning}
M.~Stern, A.~J. Liu, and V.~Balasubramanian, ``Physical effects of learning,'' {\em Physical Review E}, vol.~109, p.~024311, 2024.

\bibitem{OConnor2019SpikingEqProp}
P.~O'Connor, E.~Gavves, and M.~Welling, ``Training a spiking neural network with equilibrium propagation,'' in {\em Proceedings of the Twenty-Second International Conference on Artificial Intelligence and Statistics (AISTATS)} (K.~Chaudhuri and M.~Sugiyama, eds.), vol.~89 of {\em Proceedings of Machine Learning Research}, pp.~1516--1523, PMLR, 2019.

\bibitem{Dillavou2022DecentralizedPhysicsDrivenLearning}
S.~Dillavou, M.~Stern, A.~J. Liu, and D.~J. Durian, ``Demonstration of decentralized physics-driven learning,'' {\em Physical Review Applied}, vol.~18, p.~014040, 2022.

\bibitem{Dillavou2023MachineLearningWithoutAProcessor}
S.~Dillavou, B.~D. Beyer, M.~Stern, A.~J. Liu, M.~Z. Miskin, and D.~J. Durian, ``Machine learning without a processor: Emergent learning in a nonlinear electronic metamaterial,'' {\em arXiv preprint arXiv:2311.00537}, 2023.

\bibitem{Altman2024CoupledLearningElasticNetworks}
L.~E. Altman, M.~Stern, A.~J. Liu, and D.~J. Durian, ``Experimental demonstration of coupled learning in elastic networks,'' {\em Physical Review Applied}, vol.~22, p.~024053, 2024.

\bibitem{wanjura2025quantum}
C.~C. Wanjura and F.~Marquardt, ``Quantum equilibrium propagation for efficient training of quantum systems based on onsager reciprocity,'' {\em Nature Communications}, vol.~16, no.~1, p.~6595, 2025.

\bibitem{massar2025equilibrium}
S.~Massar and B.~M. Mognetti, ``Equilibrium propagation: the quantum and the thermal cases,'' {\em Quantum Studies: Mathematics and Foundations}, vol.~12, no.~1, p.~6, 2025.

\bibitem{scellier2024quantum}
B.~Scellier, ``Quantum equilibrium propagation: Gradient-descent training of quantum systems,'' {\em arXiv preprint arXiv:2406.00879}, 2024.

\bibitem{kitaev2002classical}
A.~Y. Kitaev, A.~H. Shen, and M.~N. Vyalyi, {\em Classical and Quantum Computation}, vol.~47 of {\em Graduate Studies in Mathematics}.
\newblock Providence, RI: American Mathematical Society, 2002.

\bibitem{kempe2006complexity}
J.~Kempe, A.~Kitaev, and O.~Regev, ``The complexity of the local {Hamiltonian} problem,'' {\em SIAM Journal on Computing}, vol.~35, no.~5, pp.~1070--1097, 2006.

\bibitem{oliveira2005complexity}
R.~Oliveira and B.~M. Terhal, ``The complexity of quantum spin systems on a two-dimensional square lattice,'' {\em Quantum Information \& Computation}, vol.~8, no.~10, pp.~900--924, 2008.

\bibitem{chen2025efficient}
C.-F. Chen, M.~Kastoryano, F.~G. Brand{\~a}o, and A.~Gily{\'e}n, ``Efficient quantum thermal simulation,'' {\em Nature}, vol.~646, no.~8085, pp.~561--566, 2025.

\bibitem{bergamaschi2026fast}
T.~Bergamaschi and C.-F. Chen, ``Fast mixing of quantum spin chains at all temperatures,'' in {\em Proceedings of the 58th Annual ACM Symposium on Theory of Computing}, pp.~835--844, 2026.

\bibitem{rouze2026efficient}
C.~Rouz{\'e}, D.~Stilck~Fran{\c{c}}a, and {\'A}.~M. Alhambra, ``Efficient thermalization and universal quantum computing with quantum {Gibbs} samplers,'' {\em Nature Physics}, vol.~22, pp.~1123--1127, 2026.

\bibitem{rouze2026optimal}
C.~Rouz{\'e}, D.~Stilck~Fran{\c{c}}a, and {\'A}.~M. Alhambra, ``Optimal quantum algorithm for gibbs state preparation,'' {\em Physical Review Letters}, vol.~136, no.~6, p.~060601, 2026.

\bibitem{tong2025fast}
Y.~Tong and Y.~Zhan, ``Fast mixing of weakly interacting fermionic systems at any temperature,'' {\em PRX Quantum}, vol.~6, no.~3, p.~030301, 2025.

\bibitem{roberts2021hidden}
D.~Roberts, A.~Lingenfelter, and A.~Clerk, ``Hidden time-reversal symmetry, quantum detailed balance and exact solutions of driven-dissipative quantum systems,'' {\em PRX Quantum}, vol.~2, no.~2, p.~020336, 2021.

\bibitem{roberts2023exact}
D.~Roberts and A.~A. Clerk, ``Exact solution of the infinite-range dissipative transverse-field ising model,'' {\em Physical review letters}, vol.~131, no.~19, p.~190403, 2023.

\bibitem{roberts2023competition}
D.~Roberts and A.~Clerk, ``Competition between two-photon driving, dissipation, and interactions in bosonic lattice models: An exact solution,'' {\em Physical Review Letters}, vol.~130, no.~6, p.~063601, 2023.

\bibitem{yao2025hidden}
M.~Yao, A.~Lingenfelter, R.~Belyansky, D.~Roberts, and A.~A. Clerk, ``Hidden time reversal in driven xxz spin chains: Exact solutions and new dissipative phase transitions,'' {\em Physical Review Letters}, vol.~134, no.~13, p.~130404, 2025.

\bibitem{lingenfelter2024exact}
A.~Lingenfelter, M.~Yao, A.~Pocklington, Y.-X. Wang, A.~Irfan, W.~Pfaff, and A.~A. Clerk, ``Exact results for a boundary-driven double spin chain and resource-efficient remote entanglement stabilization,'' {\em Physical Review X}, vol.~14, no.~2, p.~021028, 2024.

\bibitem{lingenfelter2026exact}
A.~Lingenfelter and A.~A. Clerk, ``Exact steady states of interacting driven dissipative fermionic systems with hidden time-reversal symmetry,'' {\em arXiv preprint arXiv:2605.10846}, 2026.

\bibitem{scellier2018generalization}
B.~Scellier, A.~Goyal, J.~Binas, T.~Mesnard, and Y.~Bengio, ``Generalization of equilibrium propagation to vector field dynamics,'' {\em arXiv preprint arXiv:1808.04873}, 2018.

\bibitem{scurria2026equilibrium}
A.~E. Scurria, D.~V. Abeele, B.~M. Mognetti, and S.~Massar, ``Equilibrium propagation for non-conservative systems,'' {\em arXiv preprint arXiv:2602.03670}, 2026.

\bibitem{cin2025training}
N.~D. Cin, F.~Marquardt, and C.~C. Wanjura, ``Training nonlinear optical neural networks with scattering backpropagation,'' {\em arXiv preprint arXiv:2508.11750}, 2025.

\bibitem{pathria2017statistical}
R.~K. Pathria, {\em Statistical Mechanics: International Series of Monographs in Natural Philosophy}, vol.~45.
\newblock Elsevier, 2017.

\bibitem{feynman1939forces}
R.~P. Feynman, ``Forces in molecules,'' {\em Physical review}, vol.~56, no.~4, p.~340, 1939.

\bibitem{lindblad1976generators}
G.~Lindblad, ``On the generators of quantum dynamical semigroups,'' {\em Communications in mathematical physics}, vol.~48, no.~2, pp.~119--130, 1976.

\bibitem{gorini1976completely}
V.~Gorini, A.~Kossakowski, and E.~C.~G. Sudarshan, ``Completely positive dynamical semigroups of n-level systems,'' {\em Journal of Mathematical Physics}, vol.~17, no.~5, pp.~821--825, 1976.

\bibitem{fazio2025many}
R.~Fazio, J.~Keeling, L.~Mazza, and M.~Schir{\`o}, ``Many-body open quantum systems,'' {\em SciPost Physics Lecture Notes}, p.~099, 2025.

\bibitem{davies1974markovian}
E.~B. Davies, ``Markovian master equations,'' {\em Communications in Mathematical Physics}, vol.~39, no.~2, pp.~91--110, 1974.

\bibitem{davies1976markovian}
E.~B. Davies, ``Markovian master equations. {II},'' {\em Mathematische Annalen}, vol.~219, no.~2, pp.~147--158, 1976.

\bibitem{alicki1976detailed}
R.~Alicki, ``On the detailed balance condition for non-{H}amiltonian systems,'' {\em Reports on Mathematical Physics}, vol.~10, no.~3, pp.~249--258, 1976.

\bibitem{kossakowski1977quantum}
A.~Kossakowski, A.~Frigerio, V.~Gorini, and M.~Verri, ``Quantum detailed balance and {KMS} condition,'' {\em Communications in Mathematical Physics}, vol.~57, no.~2, pp.~97--110, 1977.

\bibitem{guo2025designing}
J.~Guo, O.~Hart, C.-F. Chen, A.~J. Friedman, and A.~Lucas, ``Designing open quantum systems with known steady states: Davies generators and beyond,'' {\em Quantum}, vol.~9, p.~1612, 2025.

\bibitem{fagnola2007generators}
F.~Fagnola and V.~Umanita, ``Generators of detailed balance quantum markov semigroups,'' {\em Infinite Dimensional Analysis, Quantum Probability and Related Topics}, vol.~10, no.~03, pp.~335--363, 2007.

\bibitem{fagnola2010generators}
F.~Fagnola and V.~Umanit{\`a}, ``Generators of kms symmetric markov semigroups on symmetry and quantum detailed balance,'' {\em Communications in Mathematical Physics}, vol.~298, no.~2, pp.~523--547, 2010.

\bibitem{noachtar_exact_2026}
D.~D. Noachtar and A.~A. Clerk, ``Exact metastability in a class of driven-dissipative quantum many-body systems,'' June 2026.
\newblock arXiv:2606.07736 [quant-ph].

\bibitem{Note1}
Recall that an anti-linear operator $\protect \hat \Psi $ satisfies \(c \protect \hat \Psi = \protect \hat \Psi c^*\) for any complex scalar \(c\).

\bibitem{marcuzzi2014universal}
M.~Marcuzzi, E.~Levi, S.~Diehl, J.~P. Garrahan, and I.~Lesanovsky, ``Universal nonequilibrium properties of dissipative rydberg gases,'' {\em Physical review letters}, vol.~113, no.~21, p.~210401, 2014.

\bibitem{tucker_facilitating_2020}
K.~Tucker, D.~Barberena, R.~J. Lewis-Swan, J.~K. Thompson, J.~G. Restrepo, and A.~M. Rey, ``Facilitating spin squeezing generated by collective dynamics with single-particle decoherence,'' {\em Physical Review A}, vol.~102, p.~051701, Nov. 2020.

\bibitem{song_dissipation-induced_2025}
E.~Y. Song, D.~Barberena, D.~J. Young, E.~Chaparro, A.~Chu, S.~Agarwal, Z.~Niu, J.~T. Young, A.~M. Rey, and J.~K. Thompson, ``A dissipation-induced superradiant transition in a strontium cavity-{QED} system,'' {\em Science Advances}, vol.~11, p.~eadu5799, Apr. 2025.

\bibitem{xiang_switching_2025}
Y.-X. Xiang, W.~Li, Z.~Bai, and Y.-Q. Ma, ``Switching {Dynamics} of {Metastable} {Open} {Quantum} {Systems},'' May 2025.
\newblock arXiv:2505.05202 [quant-ph].

\bibitem{leppenen_quantum_2024}
N.~Leppenen and E.~Shahmoon, ``Quantum bistability at the interplay between collective and individual decay,'' Apr. 2024.
\newblock arXiv:2404.02134 [quant-ph] version: 1.

\bibitem{hinton2015distilling}
G.~Hinton, O.~Vinyals, and J.~Dean, ``Distilling the knowledge in a neural network,'' {\em arXiv preprint arXiv:1503.02531}, 2015.

\bibitem{savona2017spontaneous}
V.~Savona, ``Spontaneous symmetry breaking in a quadratically driven nonlinear photonic lattice,'' {\em Physical Review A}, vol.~96, no.~3, p.~033826, 2017.

\bibitem{lebreuilly2019stabilizing}
J.~Lebreuilly, C.~Aron, and C.~Mora, ``Stabilizing arrays of photonic cat states via spontaneous symmetry breaking,'' {\em Physical review letters}, vol.~122, no.~12, p.~120402, 2019.

\bibitem{rota2019quantum}
R.~Rota, F.~Minganti, C.~Ciuti, and V.~Savona, ``Quantum critical regime in a quadratically driven nonlinear photonic lattice,'' {\em Physical review letters}, vol.~122, no.~11, p.~110405, 2019.

\bibitem{schollwock_density-matrix_2011}
U.~Schollwöck, ``The density-matrix renormalization group in the age of matrix product states,'' {\em Annals of Physics}, vol.~326, pp.~96--192, Jan. 2011.

\bibitem{florido-llinas_product_2025}
M.~Florido-Llinàs, A.~M. Alhambra, R.~Trivedi, N.~Schuch, D.~Pérez-García, and J.~I. Cirac, ``The {Product} {Structure} of {Matrix} {Product} {States} under {Permutations},'' {\em PRX Quantum}, vol.~6, p.~040338, Nov. 2025.

\bibitem{raveh_dicke_2024}
D.~Raveh and R.~I. Nepomechie, ``Dicke states as matrix product states,'' {\em Physical Review A}, vol.~110, p.~052438, Nov. 2024.

\bibitem{gopalakrishnan_push-down_2025}
S.~Gopalakrishnan, ``Push-down automata as sequential generators of highly entangled states,'' {\em Journal of Physics A: Mathematical and Theoretical}, vol.~58, p.~055301, Feb. 2025.

\bibitem{hauschild_efficient_2018}
J.~Hauschild and F.~Pollmann, ``Efficient numerical simulations with {Tensor} {Networks}: {Tensor} {Network} {Python} ({TeNPy}),'' {\em SciPost Physics Lecture Notes}, p.~005, Oct. 2018.

\bibitem{Note2}
It is straightforward to extend to complex symmetric matrices via Takagi factorization.

\end{thebibliography}

\end{document}